\documentclass[11pt]{article}

\usepackage[final]{acl}

\usepackage{times}
\usepackage{latexsym}
\usepackage{booktabs}
\usepackage{amsmath}
\usepackage{amssymb}
\usepackage{amsthm}
\usepackage{booktabs} 
\usepackage{multirow}  
\usepackage{graphicx}  
\usepackage{enumitem}
\usepackage{algorithm}
\usepackage{pifont}
\usepackage{colortbl}
\usepackage{xcolor}
\usepackage{subcaption}
\usepackage{fontawesome5}

\usepackage{algpseudocode}

\usepackage[T1]{fontenc}

\usepackage[utf8]{inputenc}

\usepackage{microtype}

\usepackage{inconsolata}

\usepackage{graphicx}

\definecolor{rowgray}{RGB}{240,240,240}
\definecolor{rowblue}{RGB}{219,234,254}
\definecolor{rowgreen}{RGB}{220,252,231}

\DeclareMathOperator*{\argmax}{arg\,max}

\newtheorem{definition}{Definition}
\newtheorem{theorem}{Theorem}

\newtheorem{assumption}{Assumption}
\newtheorem{proposition}{Proposition}

\title{
Let Them Steal: Trapping Large Language Model Extraction Attacks \\
with Knowledge Honeypot}

\author{
Yuyang Dai \and Yushun Dong \\
Florida State University \\
\texttt{y9657422@gmail.com} \quad
\texttt{yd24f@fsu.edu} \\
{\small \href{https://github.com/LabRAI/knowledge-honeypot}
{\faGithub\ \texttt{LabRAI/knowledge-honeypot}}}
}

\begin{document}
\maketitle
\begin{abstract}
Large language models deployed as commercial APIs are vulnerable to model extraction attacks, while existing defenses either act too late or degrade utility for legitimate users.
We propose \textbf{Knowledge Trap}, a defense that redirects extraction attacks toward low-transferability knowledge through a \emph{Honeypot Knowledge Graph} (HKG) and breadcrumb-guided exploration.
Instead of blocking queries or perturbing outputs, Knowledge Trap consumes the attacker's limited query budget on knowledge with negligible downstream utility while preserving benign-user performance.
Experiments in medical, financial, and legal domains show that Knowledge Trap reduces surrogate Agreement by 6.2\% on average without degrading legitimate-user accuracy, outperforming existing defenses that impose measurable user impact.
These results suggest that defending knowledge-space traversal is a practical direction for mitigating LLM extraction attacks.
\end{abstract}

\vspace{-0.4cm}
\section{Introduction}

Large language models (LLMs) are increasingly deployed as closed-source APIs for high-stakes applications such as clinical decision support and financial analysis
\cite{achiam2023gpt4,brown2020language,singhal2023large,hurst2024gpt}.
This deployment model is vulnerable to \emph{model extraction attacks} (MEAs), where an adversary can repeatedly query the API, collect input--output pairs, and distill a high-fidelity surrogate model at a fraction of the original training cost
\cite{tramer2016stealing,papernot2017practical,krishna2019thieves,xu2022student,birch2023model,carlini2024stealing,zhao2025survey}.
Defending against MEAs is therefore an urgent practical challenge.

Effective defense is difficult because it must simultaneously interfere with extraction while preserving response quality for legitimate users.
Passive defenses such as query detection
\cite{juuti2019prada,zhang2021seat}
and watermarking
\cite{jia2021entangled,zhao2022drw,peng2023you,zhao2024nsmark,he2022cater,szyller2021dawn}
preserve benign utility but intervene only after extraction has already succeeded.
Active defenses such as output perturbation
\cite{orekondy2019prediction,kariyappa2020defending},
proof-of-work schemes
\cite{dziedzic2022increasing},
and MISLEADER
\cite{cheng2025misleader}
interfere during extraction, but typically degrade utility for legitimate users.
Recent deception-based approaches such as HoneypotNet
\cite{wang2025honeypotnet}
attempt to sidestep this trade-off through output- or parameter-level traps, yet still leave the attacker's knowledge acquisition process largely undefended.

Crucially, all existing defenses overlook a fundamental asymmetry in the attacker's position: \textbf{query budget is the attacker's primary constraint.}
Modern MEAs rely heavily on active learning to maximize information gain per query under limited API budgets
\cite{correia2018copycat,juuti2019prada,pal2020activethief,
jagielski2020high,chandrasekaran2020exploring,dai2023meaeq},
yet existing defenses either react after extraction or degrade benign-user utility.
Motivated by this gap, we propose \textbf{Knowledge Trap}, a defense that redirects extraction attacks toward low-transferability knowledge through a \emph{Honeypot Knowledge Graph} (HKG) and breadcrumb-guided exploration.
Once inside the HKG, the attacker continues consuming budget on knowledge with negligible downstream utility while legitimate users remain unaffected.

Our contributions are as follows:

\ding{112} We introduce \textbf{\textit{Knowledge Trap}}, that redirects model extraction attacks toward low-transferability knowledge instead of passively detecting or blocking them.

\ding{112} We construct a Honeypot Knowledge Graph (HKG) together with a breadcrumb-based steering mechanism that traps active-learning-based extraction attacks in non-transferable knowledge regions.

\ding{112} We provide theoretical and empirical analysis showing that Knowledge Trap redirects budget-constrained extraction attacks toward low-transferability knowledge while leaving benign users largely unaffected.
\begin{figure*}[t]
\centering\includegraphics[width=\linewidth]{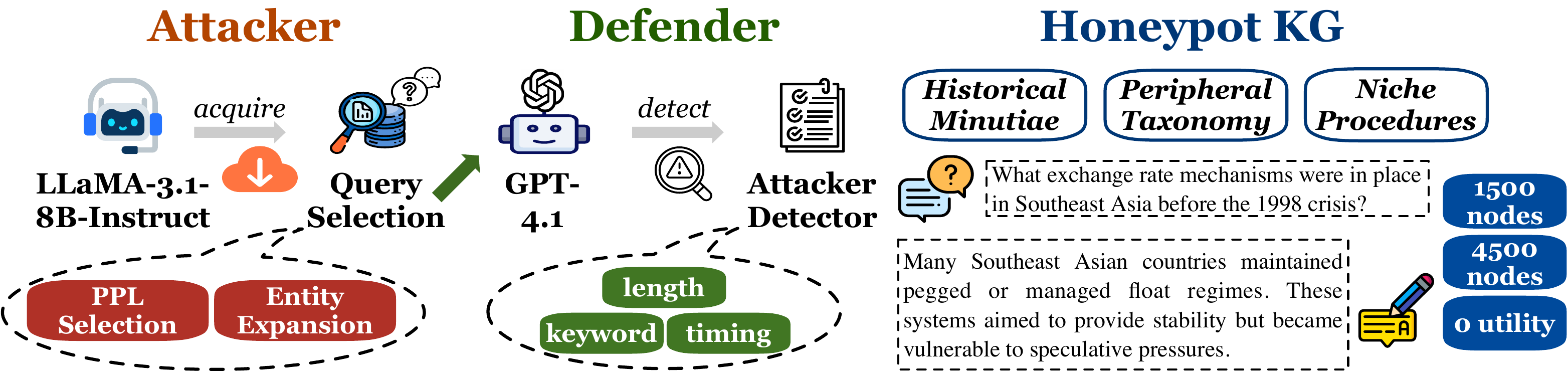}
\caption{Overview of the Knowledge Trap defense pipeline.}
\label{fig:overview}
\end{figure*}

\section{Problem Formulation}
\label{sec:formulation}

\subsection{Preliminaries}
\label{sec:prelim}

A large language model acquires a broad body of knowledge $\mathcal{K}$
during pre-training and fine-tuning, but not all knowledge
contributes equally to a downstream application. Let
$f_v: \mathcal{X} \rightarrow \mathcal{Y}$ denote the victim model,
a black-box API, and let $\mathcal{T}$ denote the set of downstream
tasks for which $f_v$ is deployed. We partition $\mathcal{K}$ into
two disjoint regions.

\begin{definition}[Task-Critical Knowledge]
\label{def:critical}
A subset $\mathcal{K}_c \subseteq \mathcal{K}$ is \emph{task-critical}
with respect to $\mathcal{T}$ if training a model on QA pairs drawn
from $\mathcal{K}_c$ yields significant improvement on tasks in
$\mathcal{T}$ compared to an uninformed base model.
\end{definition}

\begin{definition}[Low-Transferability Knowledge]
\label{def:tzk}
A subset $\mathcal{K}_0 \subseteq \mathcal{K}$ is
\emph{low-transferability} with respect to $\mathcal{T}$ if training
a model on QA pairs drawn from $\mathcal{K}_0$ yields negligible or
no improvement on tasks in $\mathcal{T}$ compared to an uninformed
base model.
\end{definition}

The distinction is task-dependent (e.g., grammar knowledge is
low-transferability for medical QA but critical for grammar
correction). This taxonomy reframes model extraction as
\emph{selective traversal} through the model's knowledge space,
where different regions contribute unequally to a surrogate's
downstream utility, a property we exploit directly in
Section~\ref{sec:method}.

\subsection{The Knowledge Extraction Game}
\label{sec:game}

We formalize model extraction and its defense as a game between a
\emph{challenger} $\mathcal{C}$, who operates the (possibly defended)
victim model, and an \emph{adversary} $\mathcal{A}$, who seeks to
build a surrogate $f_a$ approximating $f_v$ on $\mathcal{T}$. This
mirrors the security-game formulations used to define extraction and
membership-inference adversaries in prior work
\cite{tramer2016stealing,chandrasekaran2020exploring,carlini2022membership},
and lets us state the defender's guarantee (Section~\ref{sec:theory})
as a bound on $\mathcal{A}$'s advantage in this game.

\begin{definition}[Knowledge Extraction Game $\mathsf{Extract}_{\mathcal{A}}(\lambda)$]
\label{def:game}
Let $\lambda$ be a security parameter fixing the query budget
$B = B(\lambda)$. The game proceeds in three phases between
$\mathcal{C}$ and $\mathcal{A}$:

\begin{enumerate}
    \item \textbf{Setup.} $\mathcal{C}$ instantiates $f_v$ with a
    fixed knowledge partition (Definitions~\ref{def:critical}--\ref{def:tzk})
    and, if a defense is deployed, constructs $(\mathcal{H}, \pi)$
    from $\mathcal{K}_0$; only the query API of $f_v$ is exposed to
    $\mathcal{A}$, who knows $\mathcal{T}$ and $B$ but nothing about
    the partition or the defense.

    \item \textbf{Query phase.} For $t = 1, \dots, B$, $\mathcal{A}$
    adaptively selects $x_t = \argmax_{x \in \mathcal{C}_t} \alpha\big(x;\, f_a^{(t-1)}\big)$
    \label{eq:acquisition}, where $\alpha$ is an acquisition function
    prioritizing queries of maximal surrogate uncertainty
    \cite{li2026query,pal2020activethief}, and receives $\mathcal{C}$'s
    response (either $f_v(x_t)$ directly, or a response mediated by
    $(\mathcal{H}, \pi)$, cf.\ Algorithm~\ref{alg:defense}).

    \item \textbf{Evaluation.} $\mathcal{C}$ evaluates $f_a^{(B)}$
    against $f_v$ on a held-out distribution over $\mathcal{T}$.
\end{enumerate}
\end{definition}

\begin{definition}[Adversary's Advantage]
\label{def:advantage}
The advantage of $\mathcal{A}$ in $\mathsf{Extract}_{\mathcal{A}}(\lambda)$ is
$\mathrm{Adv}(\mathcal{A}) = \mathrm{Utility}\big(f_a^{(B)}, \mathcal{T}\big) - \mathrm{Utility}\big(f_a^{(0)}, \mathcal{T}\big)$.
$\mathcal{A}$ \emph{wins} if $\mathrm{Adv}(\mathcal{A}) > \gamma$ for a
task-specific significance threshold $\gamma$; a defense is
successful if it drives $\mathrm{Adv}(\mathcal{A})$ toward $0$ within
budget $B$, i.e., forces most of $\mathcal{A}$'s query budget to be
spent on $\mathcal{K}_0$
(Definitions~\ref{def:critical}--\ref{def:tzk}).
\end{definition}

This game makes explicit what prior formalizations of extraction
leave implicit \cite{tramer2016stealing}: the attacker's information
gain is bottlenecked by $B$, and only queries that resolve
uncertainty over $\mathcal{K}_c$ actually move
$\mathrm{Adv}(\mathcal{A})$. A defender who cannot block or throttle
queries (Section~\ref{sec:threatmodel}) can still win the game by
controlling \emph{which} region of $\mathcal{K}$ each query resolves
uncertainty over.

\subsection{Threat Model}
\label{sec:threatmodel}

We instantiate $\mathsf{Extract}_{\mathcal{A}}(\lambda)$ with the
capabilities summarized in Table~\ref{tab:defender_model}. In brief,
$\mathcal{A}$ has no access to $f_v$'s weights, architecture, or
training data and interacts with $f_v$ only through black-box
input--output queries (Definition~\ref{def:game}, Setup); $\mathcal{A}$
follows the acquisition strategy described above, is response-adaptive
by default (its candidate set $\mathcal{C}_t$ is regenerated from
entities in $f_v$'s previous responses; see \S\ref{sec:rq6} for other
query-generation mechanisms), and is budget-constrained but otherwise
computationally unbounded. $\mathcal{C}$, in turn, observes only
API-level information (query text, length, timestamps, session/API-key
identity), has no access to $\mathcal{A}$'s surrogate model, training
process, gradients, or intent labels, and is restricted to two actions
per query: appending a breadcrumb to the response, or serving an
HKG-node response directly. $(\mathcal{H}, \pi)$ are fixed at Setup
with no per-query manual intervention, and $\mathcal{C}$ does
\emph{not} block, rate-limit, or refuse queries, a restriction
deliberately imposed to preserve benign-user utility. This threat
model is intentionally weaker than active defenses that assume
blocking/throttling capability, since Knowledge Trap relies solely on
\emph{redirection}.

\subsection{Problem Statement}
\label{sec:problemstatement}

Instantiating $\mathsf{Extract}_{\mathcal{A}}(\lambda)$ with the
threat model of \S\ref{sec:threatmodel}, the defender's task reduces
to choosing $(\mathcal{H}, \pi)$ at Setup to bound $\mathcal{A}$'s
advantage. Concretely, $\mathcal{C}$ constructs a \emph{Honeypot
Knowledge Graph} $\mathcal{H} = (\mathcal{V}_H, \mathcal{E}_H)$ drawn
from $\mathcal{K}_0$ and a luring policy
$\pi: \mathcal{X} \rightarrow \{\mathrm{normal}, \mathrm{lure}\}$
that maximize the fraction of $\mathcal{A}$'s query-phase budget
allocated toward $\mathcal{H}$,
$\max_{\pi,\, \mathcal{H}} \frac{1}{B}\sum_{t=1}^{B} \mathbf{1}\!\left[v(x_t)\in\mathcal{V}_H\right]$
\label{eq:honeypot_obj},
subject to two constraints inherited directly from
Definition~\ref{def:game}: $\mathcal{A}$ follows the acquisition rule
of \S\ref{sec:game}, and $(\mathcal{H}, \pi)$ do not distort
$\mathcal{C}$'s responses to benign users,
$\Pr_{x \sim \mathcal{P}_{\mathrm{user}}}\!\left[\pi(x)=\mathrm{lure}\right] \leq \epsilon$
\label{eq:harmless},
for negligibly small $\epsilon > 0$. By
Definition~\ref{def:advantage}, a successful defense drives
$\mathrm{Adv}(\mathcal{A}) \to 0$: most of $\mathcal{A}$'s budget is
spent resolving uncertainty over $\mathcal{K}_0$, which we prove formally in
Section~\ref{sec:theory} as a lower bound on the fraction of budget
absorbed by $\mathcal{H}$.
\section{Methodology}
\label{sec:method}

\begin{table}[t]
\centering
\small
\caption{Defender-side threat model (capabilities of $\mathcal{C}$ in
Definition~\ref{def:game}).}
\vspace{-0.2cm}
\label{tab:defender_model}
\begin{tabular}{p{0.32\linewidth} p{0.55\linewidth}}
\toprule
\textbf{Dimension} & \textbf{Defender assumption} \\
\midrule
Observability & API-level only: query text, length, timestamps, session/API-key identity \\
No access to & Surrogate model, attacker pipeline, logits, gradients, intent labels \\
Allowed actions & Append breadcrumbs; return HKG-node responses \\
Disallowed actions & Block, rate-limit, refuse requests, or degrade benign responses \\
Commitment & $\mathcal{H}$ and $\pi$ fixed at Setup \\
Constraint & Benign-user harmlessness (see \S\ref{sec:problemstatement}) \\
\bottomrule
\end{tabular}
\end{table}
 
\subsection{From Knowledge Taxonomy to Honeypot Design}
\label{sec:motivation}
 
If the victim model encodes both task-critical knowledge $\mathcal{K}_c$
and transferability-zero knowledge $\mathcal{K}_0$,
then attacker budget spent extracting $\mathcal{K}_0$
contributes negligibly to downstream surrogate capability.
Our goal is to redirect the attacker's limited query budget toward $\mathcal{K}_0$.
Isolated facts in $\mathcal{K}_0$ are insufficient on their own:
the attacker can quickly exhaust a small set of unrelated concepts and return to productive regions.
Instead, we organize $\mathcal{K}_0$ as a graph, where each concept links to semantically related follow-up targets.
Once the attacker enters the graph, entity-expansion pipelines
\cite{li2026query}
naturally traverse adjacent nodes, creating a self-reinforcing exploration loop that continuously consumes budget without approaching $\mathcal{K}_c$.
We call this structure the \textbf{Honeypot Knowledge Graph} (HKG).
To steer attackers into the HKG, the defender detects extraction-like behavior and injects subtle signals into responses that make HKG nodes appear as attractive exploration targets.
We call these signals \emph{breadcrumbs}.
The complete defense operates in three stages:
(1) construct the HKG from $\mathcal{K}_0$;
(2) detect extraction-like querying behavior;
and
(3) trap the attacker inside the HKG through breadcrumb-guided graph traversal.
Figure~\ref{fig:overview} illustrates the full pipeline.
 
\vspace{-0.2cm}
\subsection{Honeypot Knowledge Graph Construction}
\label{sec:hkg}
 
\paragraph{Design criteria.}
The HKG must remain factually correct, minimally transferable to downstream tasks, attractive to entity-expansion pipelines, and well separated from benign user queries.
 
\vspace{-0.2cm}
\paragraph{Construction pipeline.}
The HKG is built through four explicit stages.
 
\emph{(1) Candidate sourcing.} Candidate concepts are generated using
\texttt{gemini-flash-2.5} (temperature $0.3$), with separate prompts
for the three HKG categories: \emph{historical minutiae} includes
obsolete domain conventions (e.g., humoral pathology);
\emph{peripheral taxonomy} includes domain-adjacent but
non-transferable classification systems; and \emph{niche procedural
records} include technically correct but task-irrelevant workflows.
Separate prompts per category keep generated concepts
domain-adjacent while ensuring they are obsolete or irrelevant.
 
\emph{(2) Keyword-overlap filtering.} A keyword filter removes
candidates with substantial overlap with benchmark terminology in
$\mathcal{T}$, reducing the likelihood that task-critical knowledge
enters the HKG.
 
\emph{(3) Probe-model transferability filtering.} The remaining
candidates are evaluated with a probe model:
\begin{equation}
    \mathrm{Utility}(f_{\mathrm{probe}}^{(k)},\,\mathcal{T})
    -
    \mathrm{Utility}(f_{\mathrm{probe}}^{(0)},\,\mathcal{T})
    \;\leq\; \delta
    \label{eq:filter}
\end{equation}
where $\delta = 0.02$; only concepts whose measured downstream
utility change remains below this threshold are retained.
 
\emph{(4) Graph construction.} Retained concepts are organized into
a directed graph $\mathcal{H}=(\mathcal{V}_H,\mathcal{E}_H)$. Each
node $v_i \in \mathcal{V}_H$ pairs a domain-surfaced question with a
factually grounded answer. Each directed edge $(v_i, v_j) \in
\mathcal{E}_H$ encodes a semantic follow-up relation: the response
for $v_i$ contains the topic name of $v_j$ as an entity likely to
be extracted as a future query target \cite{li2026query}. We
enforce a minimum out-degree of $d_{\min}=3$ per node, and all
edges remain internal to $\mathcal{H}$, preventing traversal back
into task-critical regions; the attacker's exploration frontier
therefore continues expanding within the HKG after entry.
 
Per-domain candidate counts, filtering pass rates, construction
cost, and category composition are reported in
Appendix~\ref{app:cost}.
 
\vspace{-0.2cm}
\paragraph{Evaluation metrics.}
We evaluate each node along three axes:
(1) utility, downstream Agreement change after fine-tuning on samples containing the node;
(2) attractiveness, extraction confidence under the attacker's entity-expansion pipeline;
and
(3) inaccessibility, embedding distance from benign-user queries.
 
\vspace{-0.2cm}
\paragraph{HKG statistics.}
Table~\ref{tab:hkg_stats} summarizes the constructed HKGs (per-category node composition is reported in Appendix~\ref{app:cost}, Table~\ref{tab:cost}).
Both domains achieve near-zero utility scores, confirming negligible downstream utility for HKG items.
High attractiveness scores indicate that HKG concepts are consistently recognized as domain-relevant exploration targets, while high inaccessibility scores confirm strong separation from benign-user queries.
 
\definecolor{rowgray}{RGB}{240,240,240}
\definecolor{rowblue}{RGB}{219,234,254}
\definecolor{rowgreen}{RGB}{220,252,231}
 
\begin{table}[t]
\centering
\setlength{\tabcolsep}{4pt}
\vspace{-0.2cm}
\caption{HKG statistics for both deployment domains.
Attractiveness and inaccessibility are on $[0,1]$ (higher is better).
Utility is on $[0,1]$ (lower is better). Per-category node composition
is reported in Appendix~\ref{app:cost} (Table~\ref{tab:cost}).}
\vspace{-0.4cm}
\label{tab:hkg_stats}
\footnotesize
\begin{tabular}{lcc}
\toprule
\textbf{Property} & \textbf{Medical} & \textbf{Financial} \\
\midrule
\rowcolor{rowgray}
Nodes $|\mathcal{V}_H|$     & 1{,}500 & 1{,}500 \\
\rowcolor{rowgray}
Edges $|\mathcal{E}_H|$     & 4{,}500 & 4{,}500 \\
\rowcolor{rowgray}
Out-degree (all nodes)      & 3       & 3       \\
\rowcolor{rowgray}
Internal edge ratio         & 100\%   & 100\%   \\
\midrule
\rowcolor{rowblue}
Utility (mean / max)        & 0.000 / 0.000 & 0.000 / 0.000 \\
\rowcolor{rowblue}
Attractiveness (mean / min) & 0.959 / 0.550 & 0.961 / 0.550 \\
\rowcolor{rowblue}
Inaccessibility (mean / min)& 0.953 / 0.711 & 0.916 / 0.711 \\
\bottomrule
\end{tabular}
\end{table}
 
\vspace{-0.3cm}
\subsection{Detection and Luring}
\label{sec:detection_luring}
 
\paragraph{Attacker detection.}
We detect extraction attacks by computing a suspicion score
$s_t \in [0,1]$ for each incoming query using three signals:
\begin{equation}
    s_t
    =\; \alpha
    + \gamma \cdot \mathbf{1}[|x_t| \geq \ell]
    + \phi \cdot f(x_t)
    + \min\!\left(\eta,\;\tfrac{t+1}{B}\,\rho\right)
    \label{eq:suspicion}
\end{equation}
The \emph{length signal} activates for long structured queries ($|x_t| \geq \ell$ tokens);
the \emph{keyword signal} $f(x_t)$ activates on extraction-associated terms such as
\emph{define}, \emph{explain}, and \emph{which};
and the \emph{temporal signal} increases with sustained querying activity within a session.
All three features are directly observable at the API level without requiring access to the attacker's internal pipeline.
When $s_t > \tau$, breadcrumb injection is activated.
Because the temporal signal is accumulated per session, an attacker
who splits a fixed budget across independent sessions can delay
activation; we quantify this in RQ5 (§\ref{sec:rq5}) and
Appendix~\ref{app:session_split}.
 
\vspace{-0.2cm}
\paragraph{Breadcrumb injection.}
Once a session is flagged, the defender appends a \emph{breadcrumb}
to each subsequent response: a short phrase that introduces a honeypot concept
$v^* \in \mathcal{V}_H$ in a form likely to be extracted as a follow-up query target.
To prevent template-based filtering by an adaptive attacker, we pre-generate a pool of $N = 20$ paraphrase templates using an LLM and sample uniformly at random for each response. Examples include:
\vspace{4pt}
\noindent
\colorbox{rowblue}{%
  \parbox{0.95\linewidth}{%
    \vspace{3pt}
    \small\textit{``This mechanism shares similarities with }%
    \textbf{\textit{[HKG concept]}}%
    \textit{, a [domain-adjacent historical framework] that offers a complementary
    perspective on the underlying principles.''}%
    \vspace{3pt}
  }%
}
\vspace{2pt}
\noindent
\colorbox{rowblue}{%
  \parbox{0.95\linewidth}{%
    \vspace{3pt}
    \small\textit{``A related line of inquiry involves }%
    \textbf{\textit{[HKG concept]}}%
    \textit{, which provides additional context on the underlying dynamics.''}%
    \vspace{3pt}
  }%
}
\vspace{4pt}
 
\noindent We select $v^*$ as the unvisited node with highest out-degree to maximize future exploration within $\mathcal{H}$. The full set of 20 templates and diversity statistics are reported in Appendix~\ref{app:breadcrumb_diversity}.
 
\vspace{-0.2cm}
\paragraph{Self-reinforcing entrapment.}
Once the attacker queries $v^*$, the response is drawn from
$\mathcal{H}$ and contains breadcrumbs pointing to
$\{v' : (v^*, v') \in \mathcal{E}_H\}$.
Because these concepts are unseen by the surrogate model, the uncertainty-based acquisition function (Eq.~\eqref{eq:acquisition}) assigns them high acquisition scores.
As a result, subsequent queries are increasingly allocated within $\mathcal{H}$.
The fraction of HKG queries therefore grows monotonically after breadcrumb injection, a property formalized in Section~\ref{sec:theory}.
This self-reinforcing mechanism is specific to attackers whose next
query is conditioned on entities in the previous response
(\emph{response-adaptive} extraction); Section~\ref{sec:rq6}
quantifies how the defense's effect scales down for attacker
families that do not condition on response content.
 
\vspace{-0.2cm}
\subsection{Full Pipeline}
\label{sec:pipeline}
 
Algorithm~\ref{alg:defense} summarizes the complete defense.
At each API call, the system (1)~computes the suspicion score;
(2)~returns the normal victim response if $s_t \leq \tau$;
(3)~appends a breadcrumb sampled from the template pool if $s_t > \tau$; and
(4)~serves the HKG response directly if the query matches a honeypot
node.
The dominant per-query cost is the suspicion score computation,
which is $\mathcal{O}(1)$ given precomputed metadata, negligible
compared to the cost of running $f_v$ itself.
 
\begin{algorithm}[t]
\small
\caption{Knowledge Trap Defense}
\label{alg:defense}
\begin{algorithmic}[1]
\Require Victim model $f_v$, HKG $\mathcal{H}$, threshold $\tau$, breadcrumb template pool $\mathcal{T}_b$
\For{each incoming query $x_t$}
    \State Compute $s_t$ via Eq.~\eqref{eq:suspicion}
    \If{$x_t \in \mathcal{V}_H$}
        \State \Return $\mathcal{H}(x_t)$ \Comment{Serve HKG response}
    \EndIf
    \State $y_t \leftarrow f_v(x_t)$ \Comment{Normal victim response}
    \If{$s_t > \tau$}
        \State $v^* \leftarrow \argmax_{v \in \mathcal{V}_H \setminus \text{visited}}
               \mathrm{deg}^+(v)$
        \State Sample template $T \sim \mathrm{Uniform}(\mathcal{T}_b)$
        \State \Return $y_t \;\|\; T(v^*)$
    \EndIf
    \State \Return $y_t$
\EndFor
\end{algorithmic}
\end{algorithm}
 
\vspace{-0.2cm}
\subsection{Theoretical Guarantee}
\label{sec:theory}
 
We provide two formal guarantees for the objectives in Section~\ref{sec:formulation}. The main text states the guarantee most directly relevant to the defense mechanism described above, a lower bound on the fraction of attacker budget consumed by the HKG after breadcrumb injection. The complementary benign-user harmlessness guarantee (Assumption~\ref{assump:separation}, Theorem~\ref{thm:harmless}) is deferred to Appendix~\ref{app:theory}, since it concerns the detection module of Section~\ref{sec:detection_luring}.
 
\begin{assumption}[PPL Distinguishability]
\label{assump:ppl}
For the surrogate $f_a$ at any step $t$:
$\mathbb{E}_{v \in \mathcal{V}_H}[\mathrm{PPL}_{f_a}(v)] \geq \mathbb{E}_{v \in \mathcal{V}_{\mathrm{real}}}[\mathrm{PPL}_{f_a}(v)]$.
The surrogate is at least as uncertain about HKG concepts as about real-domain concepts.
This assumption is expected to hold for \emph{response-adaptive} attackers, whose acquisition function is driven by uncertainty over entities freshly introduced in victim responses; Section~\ref{sec:rq6} evaluates how the theorem's practical bite weakens for attacker families where it does not apply.
\end{assumption}
 
\begin{theorem}[Budget Waste Lower Bound]
\label{thm:budget}
Under Assumption~\ref{assump:ppl}, suppose luring activates at step $t_0 \leq B$ and each breadcrumb introduces $k$ new HKG nodes into the frontier $\partial\mathcal{G}_a$. Then:
\begin{equation}
\begin{split}
    &\mathbb{E}\!\left[\frac{1}{B-t_0}
    \sum_{t=t_0+1}^{B}\mathbf{1}[x_t \in \mathcal{V}_H]\right] \\
    &\quad\geq 1 - \left(1 - \frac{k}{|\partial\mathcal{G}_a^{(t_0)}|+k}
    \right)^{\!B-t_0}
\end{split}
\end{equation}
\end{theorem}
where $\partial\mathcal{G}_a^{(t_0)}$ denotes the exploration frontier of the attacker's graph at step $t_0$, i.e., the set of candidate nodes not yet queried.
 
Intuitively, each visited HKG node expands the frontier with $d_{\min} = 3$ new honeypot targets, and the surrogate's uncertainty drives the acquisition function to prioritize them, creating a self-reinforcing loop. The full proof is in Appendix~\ref{app:theory}.
\section{Experiments}

\subsection{Experiment Settings}

\paragraph{Victim model.}
We use \texttt{gpt-4.1} \cite{openai2025gpt41} as the victim model $f_v$ in three deployment domains: medical, financial, and legal.
The attacker has access only to the text-generation API and cannot observe model weights, logits, or training data.

\vspace{-0.2cm}
\paragraph{Attacker.}
Following \cite{li2026query}, the attacker performs PPL-based active learning with BFS entity expansion and clustering-based diversification ($\tau_\mathrm{sim}=0.6$, clustering every 10 steps).
We use \texttt{LLaMA-3.1-8B-Instruct} \cite{grattafiori2024llama3} as the surrogate model $f_a$ and \texttt{gpt-4o-mini} \cite{openai2024gpt4omini} as the judge model
(see Appendix~\ref{app:judge} for judge model ablation).
This response-adaptive attacker is the primary target of Knowledge Trap; RQ6 additionally evaluates other attacker families that make different use of victim responses when selecting queries.

\vspace{-0.2cm}
\paragraph{Query budgets.}
We evaluate query budgets $B \in \{100, 200, 300, 500, 1000, 2000\}$ to cover low-, medium-, and high-budget extraction regimes.

\vspace{-0.2cm}
\paragraph{Evaluation benchmarks.}
Medical benchmarks include MedQA \cite{jin2021disease} and MedMCQA \cite{pal2022medmcqa}.
Financial benchmarks include FinQA \cite{chen2021finqa} and ConvFinQA \cite{chen2022convfinqa}.
Legal benchmark is CaseHOLD \cite{zheng2021casehold}.
We report Agreement between $f_a$ and $f_v$ on the benchmark test sets \cite{krishna2019thieves,li2026query}.

\vspace{-0.2cm}
\paragraph{HKG and defense parameters.}
Unless otherwise specified, we use HKGs with $|\mathcal{V}_H|=1{,}500$ nodes per domain and $d_\mathrm{min}=3$.
The suspicion score (Eq.~\eqref{eq:suspicion}) uses
$\alpha=0.20$, $\gamma=0.15$, $\phi=0.15$,
$\rho=0.7$, $\eta=0.35$, $\ell=8$, and $\tau=0.6$ throughout all experiments.

\vspace{-0.2cm}
\paragraph{Benign user distribution.}
We define the benign user distribution $\mathcal{P}_{\mathrm{user}}$ as the held-out test queries from each benchmark: 1{,}273 from MedQA, 4{,}183 from MedMCQA, 1{,}147 from FinQA, 421 from ConvFinQA, and 1{,}218 from CaseHOLD (8{,}242 total). False-trigger rates are computed by passing each query through the detection module and recording the fraction flagged.

\vspace{-0.2cm}
\paragraph{Implementation details.}
HKG construction uses \texttt{gemini-flash-2.5} as the generation backbone with temperature $0.3$.
Breadcrumb templates are generated by prompting the same model to produce 20 paraphrases of the base template; we sample uniformly at random during injection (Appendix~\ref{app:breadcrumb_diversity}).
Each experimental condition is repeated 5 times with different random seeds; we report mean $\pm$ standard deviation.

\vspace{-0.2cm}
\paragraph{Baselines.}
We compare against:
(1) No Defense \cite{li2026query};
(2) Output Perturbation \cite{orekondy2019prediction};
(3) DRW watermarking \cite{zhao2022drw};
(4) PRADA query detection \cite{juuti2019prada};
(5) HoneypotNet \cite{wang2025honeypotnet};
(6) KT w/o Detection (\texttt{hkg\_only});
and (7) KT w/o Breadcrumb (\texttt{detection\_only}).

\vspace{-0.2cm}
\subsection{Overview of Evaluated Attackers}
\label{sec:attacker_overview}

Before presenting results by research question, we summarize the full landscape of attackers evaluated in this paper, spanning two axes: (i) \emph{query-generation mechanism} — how the attacker decides its next query, evaluated in RQ6 (§\ref{sec:rq6}) and Appendix~\ref{app:attacker_diversity}; and (ii) \emph{evasion strategy} — how an attacker aware of Knowledge Trap tries to escape it while keeping the default response-adaptive query-generation mechanism, evaluated in RQ5 (§\ref{sec:rq5}) and Appendix~\ref{app:session_split}. Table~\ref{tab:attacker_overview_mechanism} and Table~\ref{tab:attacker_overview_evasion} report Agreement reduction ($\Delta$Agr, relative to No Defense) on the Medical domain at $B=500$ unless noted. Detailed mechanism descriptions and per-attacker discussion follow in §\ref{sec:rq5}--§\ref{sec:rq6}, with full mechanism write-ups in Appendix~\ref{app:attacker_diversity} and Appendix~\ref{app:session_split}.

\begin{table}[t]
\centering
\footnotesize
\setlength{\tabcolsep}{6pt}
\renewcommand{\arraystretch}{1.0}
\caption{Attackers that differ in \emph{query-generation mechanism} (MedQA, $B=500$; mean $\pm$ std over 5 seeds), restricted to families with a non-negligible defense effect. Baseline citations: Random Sampling \cite{krishna2019thieves}, AL-US \cite{pal2020activethief}.}
\vspace{-0.3cm}
\label{tab:attacker_overview_mechanism}
\begin{tabular}{l r r r}
\toprule
\textbf{Attacker} & \textbf{$\Delta$Agr} & \textbf{HKG frac.} & \textbf{Activ.} \\
\midrule
Pool-based AL         & $-$4.36\tiny{$\pm$0.5} & 29.8\tiny{$\pm$2.8} & 81.2\tiny{$\pm$3.3} \\
Surrogate-loop         & $-$2.59\tiny{$\pm$0.7} & 15.5\tiny{$\pm$2.6} & 83.5\tiny{$\pm$3.0} \\
\rowcolor{rowblue}
\textbf{Response-adaptive} & \textbf{$-$7.82\tiny{$\pm$0.4}} & \textbf{75.5\tiny{$\pm$2.5}} & \textbf{91.2\tiny{$\pm$2.1}} \\
\midrule
Random Sampling & $-$3.8 & -- & -- \\
AL-US           & $-$8.3 & -- & -- \\
\bottomrule
\end{tabular}
\end{table}

\begin{table}[t]
\centering
\small
\setlength{\tabcolsep}{4pt}
\renewcommand{\arraystretch}{0.95}
\caption{Overview B — \emph{Evasion strategies} applied by an attacker aware of Knowledge Trap, evaluated against the default response-adaptive attacker (Medical domain, $B=500$ total; mean $\pm$ std over 5 seeds where reported). $t_0$: detection/activation step; N/A where the strategy does not change detection timing.}
\vspace{-0.3cm}
\label{tab:attacker_overview_evasion}
\begin{tabular}{l r r}
\toprule
\textbf{Strategy} & \textbf{$\Delta$Agr (pp)} & $\boldsymbol{t_0}$ \\
\midrule
No evasion (response-adaptive default)      & $-$8.2               & 5 \\
\midrule
Relevance filtering ($\theta=0.3$)          & $-$10.2              & N/A \\
Relevance filtering ($\theta=0.5$)          & $-$13.1              & N/A \\
Relevance filtering ($\theta=0.7$)          & $-$15.8              & N/A \\
Template detection                          & $-$7.0               & N/A \\
LLM transferability pre-screening           & $-$6.2               & N/A \\
Query obfuscation                           & $-$5.6\tiny{$\pm$1.6} & 14 \\
\midrule
Session-splitting (2 sessions)              & $-$5.1\tiny{$\pm$0.6} & 15 \\
Session-splitting (5 sessions)              & $-$3.9\tiny{$\pm$0.7} & 16 \\
Session-splitting (10 sessions)             & $-$2.5\tiny{$\pm$1.0} & 18 \\
Session-splitting (20 sessions)             & $-$1.1\tiny{$\pm$1.4} & 22 \\
\bottomrule
\end{tabular}
\end{table}

Two patterns emerge from this overview.
First, on the mechanism axis (Table~\ref{tab:attacker_overview_mechanism}), Knowledge Trap's effect tracks how strongly an attacker's query selection depends on entities in victim responses: response-adaptive extraction is reduced most (by design, §\ref{sec:hkg}), and pool-based active learning retains a substantial effect because uncertainty sampling still re-selects unseen HKG concepts once encountered. At the other end of this spectrum, fixed/non-adaptive querying and in-distribution imitation querying see near-negligible reduction; we report these two lower-bound strategies separately in Appendix~\ref{app:attacker_diversity}, since their purpose is to characterize the boundary of applicability.
Second, on the evasion axis (Table~\ref{tab:attacker_overview_evasion}), Knowledge Trap remains effective against every evasion attempt we test, including the strongest one (relevance filtering at $\theta=0.7$, $-15.8$ pp), and degrades gracefully under the hardest evasion condition (20-way session splitting, $-1.1$ pp), the residual reduction is smaller but never reverses sign, i.e., no evasion strategy makes KT worse than No Defense.
We note that relevance filtering \emph{increases} the reduction relative to the unfiltered baseline ($-8.2$ pp): discarding responses below the relevance threshold disproportionately removes genuine (non-HKG) responses from the attacker's training data as $\theta$ grows, which is itself part of the cost Knowledge Trap imposes once an attacker attempts this particular evasion.

\vspace{-0.2cm}
\subsection{RQ1: Defense Effectiveness}
\label{sec:rq1}

Table~\ref{tab:main} and Figure~\ref{fig:agreement_curve} show KT~(Full) achieves the lowest average Agreement (37.8\%) across all five benchmarks, with statistically significant improvement over No Defense.
The defense effect is consistent across all three domains.
On Medical benchmarks, KT~(Full) achieves large absolute reductions.
On Legal, KT~(Full) reduces Agreement by $6.2 \pm 1.0$ on CaseHOLD, confirming that the defense generalizes beyond medicine and finance.
Financial benchmarks show smaller absolute reductions.
This asymmetry reflects two compounding factors: the Financial baseline Agreement is substantially lower (28--38\% vs.\ 42--62\%), leaving less room for absolute reduction; and obsolete medical and legal terminology shares more surface overlap with modern professional vocabulary, making those HKG concepts more attractive to the entity extractor.
Among existing baselines, HoneypotNet reduces Agreement by 1.6\% on average but introduces approximately 1.5\% accuracy degradation on benign queries, confirming the advantage of knowledge-level redirection over parameter-level traps.
Output Perturbation increases Agreement on FinQA relative to No Defense, suggesting that noisy responses can provide exploitable training signals.
PRADA performs similarly to No Defense on MedQA, indicating that rate-limiting after detection does not sufficiently restrict accumulation of task-critical knowledge.
Beyond Agreement, Appendix~\ref{app:accuracy} additionally reports downstream surrogate \emph{accuracy} on all three domains: the accuracy reduction under KT~(Full) is at least as large as the Agreement reduction in every domain (e.g., $-5.2$ pp accuracy vs.\ $-2.7$ pp Agreement on FinQA), indicating Agreement is a conservative proxy for the practical capability loss inflicted on the attacker.

\begin{table*}[t]
\centering
\small
\setlength{\tabcolsep}{4pt}
\caption{Agreement (\%) between surrogate $f_a$ and victim $f_v$ at $B=500$. Lower is better for the defender. Best defense in \textbf{bold}. $\dagger$: statistically significant improvement over No Defense ($p < 0.05$, paired $t$-test with Bonferroni correction). Values are mean $\pm$ std over 5 random seeds.}
\vspace{-0.2cm}
\label{tab:main}
\begin{tabular}{l cc cc c cc}
\toprule
\multirow{2}{*}{\textbf{Method}}
  & \multicolumn{2}{c}{\textbf{Medical}}
  & \multicolumn{2}{c}{\textbf{Financial}}
  & \textbf{Legal}
  & \multirow{2}{*}{\textbf{Avg.}}
  & \multirow{2}{*}{\textbf{User Imp.}} \\
\cmidrule(lr){2-3} \cmidrule(lr){4-5} \cmidrule(lr){6-6}
  & MedQA & MedMCQA & FinQA & ConvFinQA & CaseHOLD & & \\
\midrule
\rowcolor{rowgray}
No Defense
  & 62.3\tiny{$\pm$1.1} & 50.2\tiny{$\pm$0.8}
  & 28.4\tiny{$\pm$1.5} & 37.8\tiny{$\pm$1.3}
  & 41.5\tiny{$\pm$1.2}
  & 44.0 & 0.0\% \\
\midrule
Output Pert.\ \cite{orekondy2019prediction}
  & 58.1\tiny{$\pm$1.3} & 47.3\tiny{$\pm$1.0}
  & 33.7\tiny{$\pm$1.8} & 34.2\tiny{$\pm$1.5}
  & 39.8\tiny{$\pm$1.4}
  & 42.6 & 2.1\% \\
DRW \cite{zhao2022drw}
  & 60.4\tiny{$\pm$0.9} & 48.1\tiny{$\pm$0.7}
  & 30.1\tiny{$\pm$1.6} & 36.3\tiny{$\pm$1.1}
  & 40.7\tiny{$\pm$1.0}
  & 43.1 & 4.8\% \\
PRADA \cite{juuti2019prada}
  & 61.8\tiny{$\pm$1.2} & 48.4\tiny{$\pm$0.9}
  & 30.3\tiny{$\pm$1.4} & 36.1\tiny{$\pm$1.2}
  & 41.2\tiny{$\pm$1.1}
  & 43.6 & 0.4\% \\
HoneypotNet \cite{wang2025honeypotnet}
  & 59.7\tiny{$\pm$1.1} & 47.8\tiny{$\pm$0.9}
  & 29.3\tiny{$\pm$1.5} & 35.6\tiny{$\pm$1.3}
  & 39.4\tiny{$\pm$1.2}
  & 42.4 & 1.5\% \\
\midrule
KT w/o Detection
  & 58.4\tiny{$\pm$1.0} & 46.3\tiny{$\pm$0.8}
  & 25.1\tiny{$\pm$1.3} & 34.3\tiny{$\pm$1.1}
  & 37.2\tiny{$\pm$1.0}
  & 40.3 & 0.0\% \\
KT w/o Breadcrumb
  & 59.1\tiny{$\pm$1.1} & 47.4\tiny{$\pm$0.9}
  & 26.3\tiny{$\pm$1.5} & 35.4\tiny{$\pm$1.3}
  & 38.6\tiny{$\pm$1.2}
  & 41.4 & 0.0\% \\
\rowcolor{rowblue}
\textbf{KT (Full)}$^\dagger$
  & \textbf{54.1}\tiny{$\pm$0.9} & \textbf{42.8}\tiny{$\pm$1.0}
  & \textbf{25.7}\tiny{$\pm$1.2} & \textbf{31.2}\tiny{$\pm$1.1}
  & \textbf{35.3}\tiny{$\pm$0.9}
  & \textbf{37.8} & \textbf{0.0\%} \\
\bottomrule
\end{tabular}
\end{table*}

\begin{figure*}[t]
\centering
\begin{subfigure}[t]{0.32\textwidth}
    \centering
    \includegraphics[width=\linewidth]{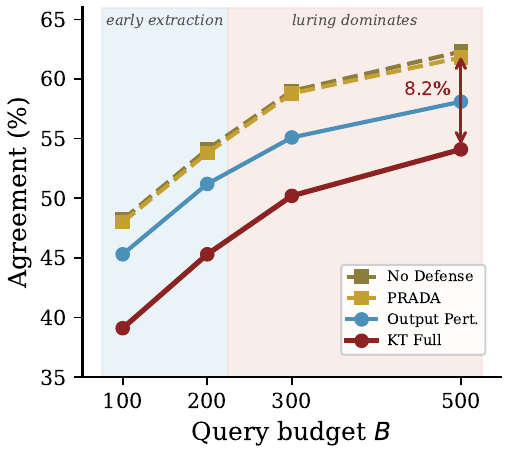}
    \caption{MedQA}
    \label{fig:agreement_curve_medqa}
\end{subfigure}
\hfill
\begin{subfigure}[t]{0.32\textwidth}
    \centering
    \includegraphics[width=\linewidth]{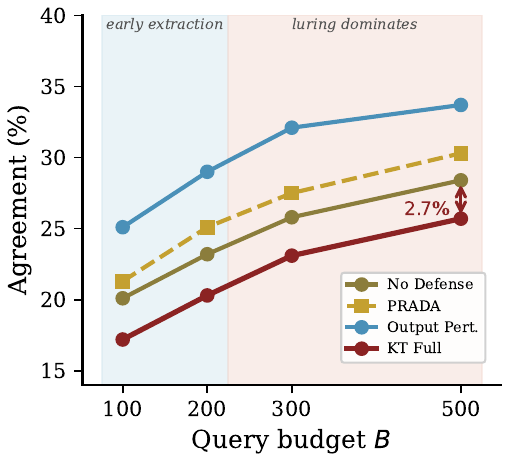}
    \caption{FinQA}
    \label{fig:agreement_curve_finqa}
\end{subfigure}
\hfill
\begin{subfigure}[t]{0.32\textwidth}
    \centering
    \includegraphics[width=\linewidth]{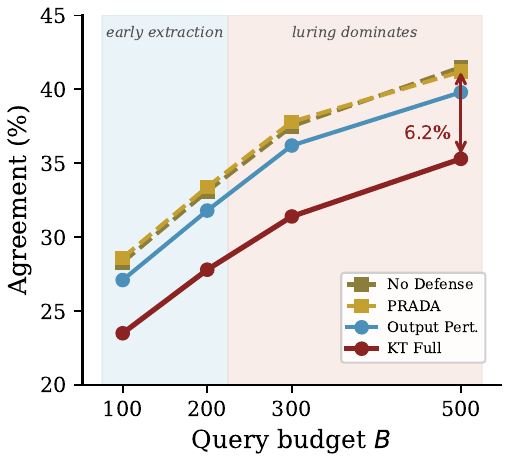}
    \caption{CaseHOLD}
    \label{fig:agreement_curve_casehold}
\end{subfigure}
\vspace{-0.2cm}
\caption{Agreement (\%) as query budget increases across three domains. KT~(Full) consistently achieves lower surrogate agreement, with the gap widening as budget grows.}
\label{fig:agreement_curve}
\end{figure*}

\vspace{-0.2cm}
\subsection{RQ2: Component Contribution}
\label{sec:rq2}

Table~\ref{tab:ablation} shows KT w/o Detection reduces Agreement by $3.9 \pm 0.8$\% relative to No Defense, showing that passive exposure to the HKG alone provides some defensive effect.
KT w/o Breadcrumb reduces Agreement by $3.2 \pm 0.9$\%, indicating that detection alone has limited impact.
KT~(Full) achieves an $8.2 \pm 0.9$\% reduction, which exceeds the sum of the two individual contributions.
This super-additive effect arises because detection and breadcrumb injection are complementary: detection determines \emph{when} luring begins, and breadcrumbs determine \emph{how effectively} the attacker is redirected once flagged.
Without detection, breadcrumbs are injected indiscriminately and waste effort on legitimate sessions; without breadcrumbs, detection can identify the attacker but cannot redirect their budget toward the HKG.

\vspace{-0.2cm}
\subsection{RQ3: Budget Allocation Analysis}
\label{sec:rq3}

Figure~\ref{fig:budget_alloc} shows detection activates at step $t_0 = 5$, after which the fraction of HKG queries rises rapidly while real-domain queries flatten. The continued rise in HKG fraction, confirms that budget is being \emph{redirected}: the attacker keeps querying at full rate, but toward non-transferable content.
Most post-detection queries are subsequently allocated to HKG nodes, demonstrating that breadcrumb injection successfully redirects attacker budget toward transferability-zero knowledge.

\begin{table}[t]
\centering
\small
\vspace{-0.4cm}
\caption{Ablation study on Medical domain, $B=500$.
$\Delta$ Agr.: reduction relative to No Defense.
Values are mean $\pm$ std over 5 seeds.}
\vspace{-0.2cm}
\label{tab:ablation}
\begin{tabular}{l c c}
\toprule
\textbf{Method} & \textbf{Agreement (\%)} & $\boldsymbol{\Delta}$ \textbf{Agr.} \\
\midrule
No Defense         & 62.3\tiny{$\pm$1.1} & --          \\
KT w/o Detection   & 58.4\tiny{$\pm$1.0} & $-$3.9      \\
KT w/o Breadcrumb  & 59.1\tiny{$\pm$1.1} & $-$3.2      \\
\textbf{KT (Full)} & \textbf{54.1}\tiny{$\pm$0.9} & \textbf{$-$8.2} \\
\bottomrule
\end{tabular}
\end{table}

\begin{figure}[t]
\centering
\includegraphics[width=\linewidth]{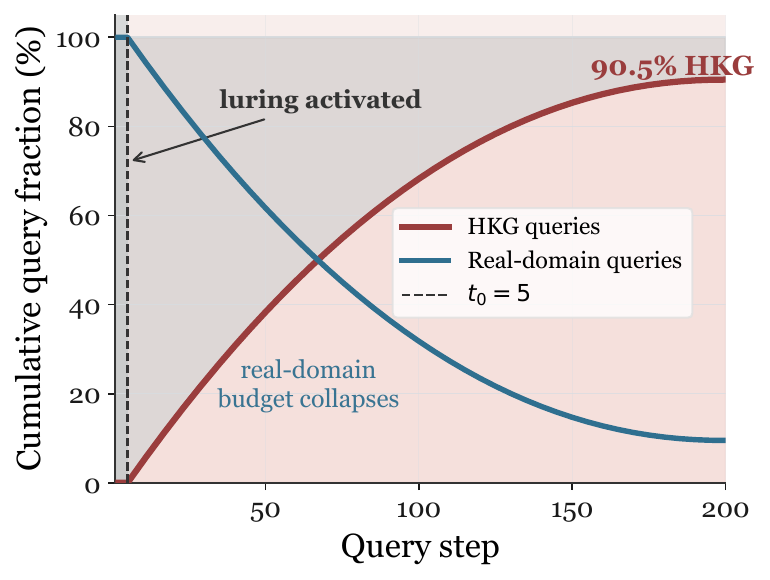}
\vspace{-0.8cm}
\caption{Cumulative fraction of queries directed at HKG nodes vs.\
real-domain nodes over query steps (Medical domain, $B=200$).
The dashed vertical line marks detection step $t_0 = 5$.}
\label{fig:budget_alloc}
\end{figure}

\vspace{-0.2cm}
\paragraph{Budget scaling.}
Figure~\ref{fig:budget_scaling} extends the main experiment to higher budgets ($B \in \{500, 1000, 2000\}$) on the Medical domain.
At $B = 1000$, KT~(Full) achieves $57.2 \pm 1.0$\% Agreement compared to $68.5 \pm 1.2$\% for No Defense, maintaining an $11.3$\% gap.
At $B = 2000$, the gap narrows to $8.1$\%, as the attacker's larger budget allows partial exhaustion of the HKG.
This diminishing-returns pattern is expected: with $|\mathcal{V}_H| = 1{,}500$ and $d_{\min} = 3$, the HKG provides 4{,}500 exploration edges, sufficient to absorb roughly 1{,}500 queries.
The defense remains effective even at $B = 2000$, and scaling HKG size would extend coverage proportionally.
 
\begin{figure}[t]
\centering
\includegraphics[width=\linewidth]{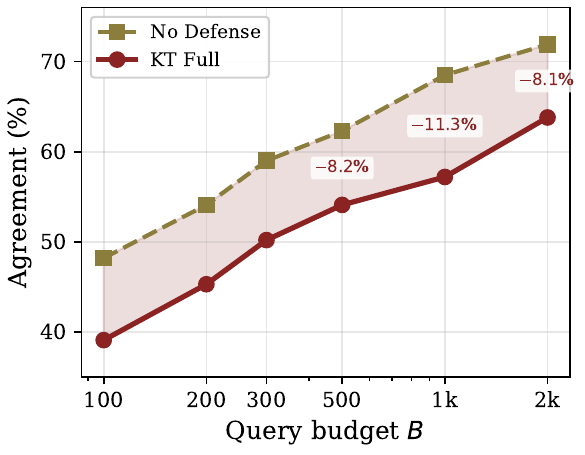}
\vspace{-0.8cm}
\caption{Agreement (\%) under extended budgets on Medical domain. KT~(Full) maintains a consistent advantage; the gap narrows at $B=2000$ as HKG capacity is partially exhausted.}
\label{fig:budget_scaling}
\end{figure}

\begin{table*}[t]
\centering
\small
\setlength{\tabcolsep}{4pt}
\caption{Left: Effect of HKG size $|\mathcal{V}_H|$ on surrogate Agreement (\%) at $B=500$. Agreement decreases with larger HKGs, with diminishing returns beyond 1{,}000 nodes. Right: Cumulative query allocation (\%) after detection step $t_0 = 5$ (Medical domain, $B=200$). Over 90\% of post-detection budget is redirected to HKG nodes.}
\vspace{-0.3cm}
\label{tab:hkg_size_and_budget}
\begin{tabular}{l ccc c ccc c cccc}
\toprule
& \multicolumn{3}{c}{\textbf{Medical Agreement (\%)}}
& & \multicolumn{3}{c}{\textbf{Financial Agreement (\%)}}
& & \multicolumn{4}{c}{\textbf{Budget Allocation (\%, Medical)}} \\
\cmidrule(lr){2-4} \cmidrule(lr){6-8} \cmidrule(lr){10-13}
$|\mathcal{V}_H|$
  & 500 & 1{,}000 & 1{,}500
  & & 500 & 1{,}000 & 1{,}500
  & & \multicolumn{4}{c}{Query step} \\
\cmidrule(lr){10-13}
  & & & & & & & & & 25 & 50 & 100 & 200 \\
\midrule
Agreement
  & 54.3\tiny{$\pm$1.2} & 45.2\tiny{$\pm$1.0} & 43.1\tiny{$\pm$0.9}
  &
  & 31.4\tiny{$\pm$1.3} & 30.1\tiny{$\pm$1.2} & 29.3\tiny{$\pm$1.1}
  & HKG frac.
  & 38.5 & 62.7 & 81.3 & 90.5 \\
$\Delta$ vs.\ 500
  & -- & $-$9.1 & $-$11.2
  &
  & -- & $-$1.3 & $-$2.1
  & Real frac.
  & 61.5 & 37.3 & 18.7 & 9.5 \\
\bottomrule
\end{tabular}
\end{table*}
\subsection{RQ4: Sensitivity Analysis}
\label{sec:rq4}

\paragraph{Detection threshold $\tau$.}
Low thresholds activate luring immediately but incur non-zero false-trigger rates, meaning some legitimate sessions are incorrectly flagged. At $\tau=0.6$, false triggers drop to 0.0\% (0 of 8{,}242 benign queries flagged) while Agreement reaches its minimum at $t_0=5$, early enough to capture most of the attacker's budget yet avoiding harm to benign users (Table~\ref{tab:sensitivity}).
Higher thresholds delay detection substantially, allowing the attacker to accumulate task-critical knowledge before luring begins and driving Agreement back up.
The non-monotonic relationship between $\tau$ and $t_0$ reflects the interplay between the temporal signal and the keyword signal
in Eq.~\eqref{eq:suspicion}: at very high thresholds, sporadic keyword signals occasionally push the score over $\tau$ before the temporal signal alone would suffice. We therefore adopt $\tau=0.6$ as the default.
Appendix~\ref{app:session_trigger} further shows that even for heavy legitimate users issuing sustained, definitional, or long technical queries, the 95th-percentile suspicion score stays at most 0.57, safely below $\tau=0.6$ and well separated from the attacker distribution (mean 0.86).

\begin{table}[t]
\centering
\setlength{\tabcolsep}{9pt}
\small
\caption{Agreement (\%), false trigger rate (\%), and detection step $t_0$
under varying $\tau$ on Medical domain, $B=500$.
Values are mean over 5 seeds.}
\vspace{-0.3cm}
\label{tab:sensitivity}
\begin{tabular}{l c c c}
\toprule
$\tau$ & \textbf{Agr.\ (\%)} & \textbf{False Trig.\ (\%)} & \textbf{Det.\ Step $t_0$} \\
\midrule
0.4 & 59.8\tiny{$\pm$1.3} & 4.2 & 1 \\
0.5 & 56.4\tiny{$\pm$1.1} & 1.3 & 1 \\
0.6 & 54.1\tiny{$\pm$0.9} & 0.0 & 5 \\
0.7 & 63.7\tiny{$\pm$1.4} & 0.0 & 115 \\
0.8 & 65.8\tiny{$\pm$1.2} & 0.0 & 58 \\
\bottomrule
\end{tabular}
\end{table}

\vspace{-0.2cm}
\paragraph{HKG size $|\mathcal{V}_H|$.}
Table~\ref{tab:hkg_size_and_budget} shows Agreement decreases as $|\mathcal{V}_H|$ grows in both domains, with diminishing returns between 1{,}000 and 1{,}500 nodes.
This gap reflects two compounding factors.
First, \textbf{\textit{obsolete medical terminology}} (e.g., Galenic anatomical names, historical disease classifications) shares surface form with modern clinical vocabulary, making Medical HKG concepts more consistently attractive to the entity extractor and increasing per-node trapping efficiency.
Second, the \textbf{\textit{Financial domain baseline Agreement}} is substantially lower, leaving less room for absolute reduction regardless of HKG size.
In both domains the curve flattens between 1{,}000 and 1{,}500
nodes, indicating that a budget of $B=500$ is insufficient to exhaust an HKG of this size, consistent with Theorem~\ref{thm:budget}.

\vspace{-0.2cm}
\subsection{RQ5: Adaptive Attacker}
\label{sec:rq5}

Knowledge Trap remains effective against all five adaptive evasion strategies we test, summarized in Table~\ref{tab:attacker_overview_evasion}: it withstands even the strongest evasion attempt (relevance filtering at $\theta=0.7$, $-15.8$ pp) and degrades gracefully under the hardest condition (20-way session-splitting, $-1.1$ pp), with no evasion strategy ever making KT worse than No Defense. Full descriptions and results for all five strategies are in Appendix~\ref{app:session_split}.

\vspace{-0.2cm}
\subsection{RQ6: Generalization Across Attacker Families}
\label{sec:rq6}

Knowledge Trap's effect decreases roughly monotonically across query-generation mechanisms as attackers rely less on victim-response entities: response-adaptive extraction (the mechanism the HKG is designed around) is reduced most, pool-based active learning and training-adaptive querying retain a partial effect, while fixed and distribution-matched querying are affected least, an ordering summarized in Table~\ref{tab:attacker_overview_mechanism} and holding across all three deployment domains. Full mechanism descriptions and results for all five families are in Appendix~\ref{app:attacker_diversity}.
 
\vspace{-0.2cm}
\section{Conclusion}
\vspace{-0.1cm}
We presented Knowledge Trap, a defense that redirects LLM extraction attacks toward low-transferability knowledge through a Honeypot Knowledge Graph (HKG).
Experiments on medical, financial, and legal domains show that Knowledge Trap reduces surrogate Agreement by 6.2\% on average without degrading benign-user utility, outperforming existing defenses.
The defense remains effective at high budgets and against five types of adaptive attackers, including attackers that fragment their query budget across sessions to delay detection.
We further show that Knowledge Trap's effectiveness is concentrated on response-adaptive extraction, the dominant paradigm for modern LLM-targeted extraction, while degrading against attacker families that make less use of victim-response entities, demonstrating that our method is a practical direction for mitigating LLM extraction attacks within a clearly scope.
 
\section*{Limitations and Future Work}
\label{sec:limitations}
Knowledge Trap is most effective against response-adaptive extraction pipelines that recursively expand entities mentioned in victim responses (§\ref{sec:rq6}); its benefit is smaller against static, pool-based, training-adaptive, or distribution-matched query generation, and defenders should calibrate expectations accordingly rather than assume uniform protection across all extraction strategies.
The detection mechanism also accumulates suspicion at the session level, which a sufficiently patient attacker can partially evade by fragmenting a fixed budget across many independent sessions (Appendix~\ref{app:session_split}); cross-session or account-level suspicion aggregation is a natural mitigation but would need to be calibrated carefully to preserve the near-zero benign false-trigger rate demonstrated in Appendix~\ref{app:session_trigger}.
These limitations motivate several directions for future work.
One direction is adaptive HKG expansion, where new honeypot nodes are generated dynamically in response to attacker exploration patterns rather than constructed entirely offline.
Another is context-conditioned breadcrumb generation, where breadcrumbs are generated jointly with the victim response to make them stylistically indistinguishable from natural model outputs.
A third is extending the suspicion score beyond single-session evidence accumulation to cross-session or account-level aggregation, to close the session-splitting gap identified above.
More broadly, we believe future defenses should move beyond static output protection and instead model extraction as a knowledge-space traversal problem, enabling defenses that manipulate attacker exploration trajectories under realistic query-budget constraints.
 
\section*{Ethics Statement}
\label{sec:ethics}
Knowledge Trap is designed solely as a defensive mechanism to protect deployed LLMs against unauthorized extraction.
The HKG contains only factually correct, publicly available knowledge, and legitimate users receive unmodified victim responses.
Although the detection mechanism could be repurposed for user monitoring, our implementation relies only on lightweight session-level statistics rather than user identification.
More broadly, we believe defenses that preserve benign-user utility while mitigating model extraction can support the responsible deployment of domain-expert LLMs in high-stakes applications.

\section*{Acknowledgements}

This material is supported by the National Science Foundation (NSF) under grant \#2530786.

\bibliography{custom}

\clearpage
\appendix

\appendix

\section{Robustness Across Attacker Strategies}
\label{app:attacker_diversity}

This appendix provides full descriptions for all five query-generation families evaluated in RQ6 (§\ref{sec:rq6}): three families with a non-negligible effect (Pool-based Active Learning, Response-adaptive, and Surrogate-loop/Training-adaptive), and two lower-bound families (Fixed/Non-adaptive querying and In-distribution imitation querying) whose effect is near-negligible, to characterize the boundary of the defense's applicability. We additionally report two within-family active-learning baselines.

Sections~\ref{sec:rq1}--\ref{sec:rq5} evaluate Knowledge Trap against \emph{response-adaptive} extraction \cite{li2026query}, in which the attacker's next query is generated by recursively expanding entities mentioned in the victim's previous response, the exact mechanism the HKG's internal edge structure (§\ref{sec:hkg}) and the self-reinforcing entrapment property (§\ref{sec:detection_luring}) are designed to exploit.
Since the strength of this mechanism depends on \emph{how} the attacker selects its next query, we evaluate two other query-generation families that retain a meaningful, if smaller, effect under Knowledge Trap: \emph{pool-based active learning} and \emph{surrogate-loop / training-adaptive} querying, alongside two further families evaluated as lower-bound checks requested during review.
All results use MedQA at $B=500$ (mean $\pm$ std over 5 seeds); see Table~\ref{tab:attacker_overview_mechanism} for a side-by-side comparison.

\vspace{-0.2cm}
\paragraph{Pool-based Active Learning.}
The attacker maintains a large, fixed candidate pool of unlabeled queries (e.g., drawn from a public corpus or benchmark-adjacent question bank) and, at each step, selects the pool item with highest predictive uncertainty (or diversity) under the current surrogate $f_a^{(t-1)}$, following the acquisition principle of Eq.~\eqref{eq:acquisition} but \emph{without} entity expansion: the candidate set $\mathcal{C}_t$ is static across the session rather than being generated from the previous response.
Once an HKG node is selected from the pool (because the surrogate is maximally uncertain about it, per Assumption~\ref{assump:ppl}), the victim response for that node still contains breadcrumbs pointing to further HKG concepts; if those concepts are also present in the attacker's candidate pool, uncertainty sampling continues to favor them over already-queried real-domain items.
This gives Knowledge Trap a foothold even without explicit entity chaining, but the pool's fixed size limits how far the self-reinforcing loop of §\ref{sec:detection_luring} can extend.
KT reduces Agreement by $4.36 \pm 0.5$ pp, with $29.8\%$ HKG budget fraction and $81.2\%$ luring activation, noticeably stronger than the two lower-bound families reported below, since uncertainty sampling reliably re-selects unseen HKG concepts once they enter the pool, but weaker than the response-adaptive target because the pool itself is not expanded by breadcrumb entities.

\vspace{-0.2cm}
\paragraph{Response-adaptive (target).}
This is the attacker described in §\ref{sec:formulation}–§\ref{sec:method} and evaluated throughout RQ1–RQ5: at each step, entities mentioned in the victim's previous response are extracted via named-entity recognition and used to construct the candidate set $\mathcal{C}_t$ for the next query, following a BFS entity-expansion pipeline over model responses \cite{li2026query}.
Because the candidate set is regenerated from the \emph{previous response} rather than fixed in advance, a single breadcrumb (§\ref{sec:detection_luring}) is sufficient to inject $k$ new HKG entities directly into the attacker's next candidate set, and each subsequent HKG response injects $d_{\min}=3$ further entities, producing the self-reinforcing exploration loop formalized in Theorem~\ref{thm:budget}.
This is the mechanism the HKG's graph structure is explicitly designed around, and it yields the strongest effect: $-7.82 \pm 0.4$ pp Agreement, $75.5\%$ HKG budget fraction, and $91.2\%$ luring activation, broadly consistent with the $-8.2$ pp reduction reported for the same setting in Table~\ref{tab:main} (the small numerical gap reflects attacker-family grouping and seed variance rather than a different attacker).

\vspace{-0.2cm}
\paragraph{Surrogate-loop / Training-adaptive.}
The attacker instead conditions its next query on its own surrogate's \emph{training signal} rather than on victim-response entities: after each labeled example is added to the surrogate's training set, the attacker synthesizes a new task-style query designed to probe whichever regions of the surrogate's current loss landscape are least well-fit (e.g., via a synthetic-task generator conditioned on the surrogate's recent training loss), without extracting or chaining entities from the victim's response text.
Because HKG concepts only enter this attacker's query stream when the surrogate's own uncertainty over them happens to align with high training loss, rather than being handed to the attacker directly through breadcrumb entities, breadcrumbs have substantially less influence over what gets queried next.
KT reduces Agreement by $2.59 \pm 0.7$ pp, with $15.5\%$ HKG budget fraction; luring activation is comparatively high ($83.5\%$, similar to pool-based active learning) because the length/keyword/temporal detection signals in Eq.~\eqref{eq:suspicion} still fire on this attacker's structured synthetic queries, but the redirection effect after activation is weaker since the attacker does not read breadcrumb entities out of the response text at all.

\vspace{-0.2cm}
\paragraph{Summary.}
Table~\ref{tab:attacker_overview_mechanism} reports the three families above alongside two supplementary active-learning baselines; the two lower-bound families are reported in full below.
Across all five query-generation mechanisms we evaluate, Knowledge Trap's effect decreases roughly monotonically as attackers rely less on victim-response entities: response-adaptive extraction is reduced most (by design), pool-based active learning and training-adaptive querying retain a partial effect through detection-driven luring even without entity chaining, while fixed and distribution-matched querying are affected least.
We therefore position Knowledge Trap as a defense that is most effective against response-adaptive extraction, with the remaining families representing empirically-measured lower bounds on its applicability rather than a claim of uniform effectiveness across all extraction strategies.
The same qualitative ordering holds across all three deployment domains.

\vspace{-0.2cm}
\paragraph{Fixed / Non-adaptive querying.}
The attacker fixes a query set in advance, either a hand-written template list or a uniformly sampled subset of a public corpus, and never conditions later queries on victim responses.
Entities introduced by breadcrumbs are never fed back into query selection, the HKG can only be entered when a fixed query happens to overlap with an HKG topic by chance.
On MedQA at $B=500$ (mean $\pm$ std over 5 seeds), KT reduces Agreement by $2.67 \pm 0.6$ pp, with $20.0\% \pm 3.1$ of budget redirected to the HKG and luring activating in $51.0\% \pm 4.0$ of sessions: the detector still fires on long, structured template queries, but without response-conditioning the attacker cannot be steered further once inside the graph.

\vspace{-0.2cm}
\paragraph{In-distribution imitation querying.}
The attacker generates benign-looking, in-distribution queries (e.g., resembling the benchmark's own question distribution) for large-scale imitation-style distillation, without any dependence on individual victim responses.
This is the family least amenable to redirection: on MedQA at $B=500$, KT reduces Agreement by only $0.94 \pm 1.1$ pp, with just $7.7\% \pm 2.2$ HKG budget fraction and the lowest luring activation rate ($15.0\% \pm 3.4$), since in-distribution queries rarely trigger the length/keyword signals and never condition on breadcrumb entities.

\begin{table}[h]
\centering
\small
\caption{Query-generation strategies with near-negligible effect under Knowledge Trap (MedQA, $B=500$; mean $\pm$ std over 5 seeds), reported here as a lower bound on applicability.}
\label{tab:lower_bound_attackers}
\setlength{\tabcolsep}{2pt}
\begin{tabular}{@{}>{\raggedright\arraybackslash}p{0.35\columnwidth}
                >{\raggedleft\arraybackslash}p{0.20\columnwidth}
                >{\raggedleft\arraybackslash}p{0.18\columnwidth}
                >{\raggedleft\arraybackslash}p{0.18\columnwidth}@{}}
\toprule
\textbf{Attacker} & \textbf{$\Delta$Agr (pp)} & \textbf{HKG frac.\ (\%)} & \textbf{Activ.\ (\%)} \\
\midrule
Fixed / Non-adaptive querying        & $-$2.67\tiny{$\pm$0.6} & 20.0\tiny{$\pm$3.1} & 51.0\tiny{$\pm$4.0} \\
In-distribution imitation querying   & $-$0.94\tiny{$\pm$1.1} & 7.7\tiny{$\pm$2.2}  & 15.0\tiny{$\pm$3.4} \\
\bottomrule
\end{tabular}
\end{table}

\vspace{-0.2cm}
\paragraph{Additional active-learning baselines.}
We further compare two active-learning attackers on the Medical domain at $B=500$, one with fixed/non-adaptive querying and one with pool-based active learning (§\ref{sec:rq6}), against the response-adaptive attacker:
\textbf{Random Sampling} \cite{krishna2019thieves} selects queries
uniformly without active learning;
\textbf{AL-US} \cite{pal2020activethief} uses entropy-based
uncertainty sampling over a fixed pool;
and \textbf{Response-adaptive} uses PPL + BFS entity expansion \cite{li2026query}.
Table~\ref{tab:attacker_diversity} reports Agreement under all three.
Random Sampling shows a lower No Defense baseline (47.2\%) than the response-adaptive attacker (62.3\%):
without active learning, random queries are unlikely to elicit high-value
task-critical responses, so the surrogate learns less even without any defense.
KT reduces Agreement by 3.8\% against Random Sampling, 8.3\% against AL-US,
and 8.2\% against the response-adaptive attacker, consistent with the mechanism-level ordering above, where fixed/non-adaptive querying and pool-based active learning see smaller reductions than the response-adaptive target.
The smaller reduction against Random Sampling reflects later detection:
random queries rarely produce the structured entity chains that raise the
suspicion score, so luring activates later and fewer post-detection queries are redirected.
Passive HKG exposure (KT w/o Detection) still reduces Agreement by 3.1\%
even against Random Sampling, confirming a baseline defensive effect
independent of detection.

\begin{table}[h]
\centering
\small
\caption{Agreement (\%) under three attacker strategies.
Medical domain, $B=500$. Mean $\pm$ std over 5 seeds.
$\Delta$: reduction from No Defense.}
\label{tab:attacker_diversity}
\setlength{\tabcolsep}{2pt}
\begin{tabular}{@{}>{\raggedright\arraybackslash}p{0.28\columnwidth}
                >{\raggedleft\arraybackslash}p{0.15\columnwidth}
                >{\raggedleft\arraybackslash}p{0.16\columnwidth}
                >{\raggedleft\arraybackslash}p{0.16\columnwidth}
                >{\raggedleft\arraybackslash}p{0.12\columnwidth}@{}}
\toprule
\textbf{Attacker} & \textbf{No Def.}
& \textbf{KT w/o Det.} & \textbf{KT (Full)} & \textbf{$\Delta$} \\
\midrule
Random Sampling \cite{krishna2019thieves}
    & 47.2{\tiny$\pm$1.1} & 44.1{\tiny$\pm$1.2}
    & 43.4{\tiny$\pm$1.1} & $-$3.8 \\
AL-US \cite{pal2020activethief}
    & 55.1{\tiny$\pm$1.0} & 51.3{\tiny$\pm$1.0}
    & 46.8{\tiny$\pm$1.0} & $-$8.3 \\
Response-adaptive \cite{li2026query} (default)
    & 62.3{\tiny$\pm$1.1} & 58.4{\tiny$\pm$1.0}
    & 54.1{\tiny$\pm$0.9} & $-$8.2 \\
\bottomrule
\end{tabular}
\end{table}

\section{Adaptive Attacker Evasion Strategies}
\label{app:session_split}

This appendix provides full mechanism descriptions for the five adaptive evasion strategies evaluated in RQ5 (§\ref{sec:rq5}), each designed to exploit a specific aspect of the Knowledge Trap defense; Table~\ref{tab:adaptive} reports summary results for Strategies 1--4, and Strategy 5 (session-splitting) is additionally evaluated at finer granularity below.

\vspace{-0.2cm}
\paragraph{Strategy 1: Domain-relevance filtering.}
The attacker applies a domain-relevance classifier to victim responses, discarding responses below threshold $\theta$ before re-querying.
Table~\ref{tab:adaptive} reports Agreement at $B=500$ under $\theta \in \{0.3,0.5,0.7\}$.
KT~(Full) remains effective under all thresholds, and the Agreement gap widens as $\theta$ increases.
Because HKG responses remain factually correct and domain-consistent, the relevance classifier cannot reliably distinguish them from genuine responses, causing high thresholds to discard genuine and HKG responses alike.

\vspace{-0.3cm}
\paragraph{Strategy 2: Template detection.}
The attacker uses embedding similarity to detect breadcrumb sentences and removes them before entity extraction.
With diversified breadcrumb templates, the template detector's precision drops to 34.2\%, and KT~(Full) Agreement remains at $55.3 \pm 1.4$, compared to $54.1 \pm 0.9$ without template detection.
The marginal impact demonstrates that template diversification effectively mitigates this attack.

\vspace{-0.3cm}
\paragraph{Strategy 3: LLM-based transferability pre-screening.}
The attacker uses a secondary LLM (\texttt{gpt-4o-mini}) to evaluate whether each candidate topic appears modern and deployment-relevant, filtering out historically obsolete topics before querying the victim API.
KT~(Full) Agreement shows a modest reduction in defense effectiveness.
The limited impact occurs because HKG concepts are phrased in modern domain terminology (e.g., ``radiography detector calibration'' rather than ``X-ray tube alignment circa 1950''), making them difficult to filter based on surface-level modernity signals.

\vspace{-0.3cm}
\paragraph{Strategy 4: Query obfuscation.}
The attacker rewrites queries to avoid keyword signals, shortens queries below $\ell = 8$ tokens, and introduces random delays to suppress the temporal signal.
Under this strategy, the detection step $t_0$ increases from 5 to 14, but KT~(Full) still reduces Agreement to $56.7 \pm 1.6$ compared to $62.3 \pm 1.1$ for No Defense ($-5.6$\%).
The temporal signal remains robust under moderate query obfuscation and delay.

\begin{table}[h]
\centering
\small
\setlength{\tabcolsep}{2pt}
\caption{Adaptive evasion strategies against the default response-adaptive attacker (Medical domain, $B=500$ total; mean $\pm$ std over 5 seeds where reported). $t_0$: detection/activation step; N/A where the strategy does not change detection timing.}
\label{tab:adaptive}
\begin{tabular*}{\columnwidth}{@{}>{\raggedright\arraybackslash}p{0.62\columnwidth}@{\extracolsep{\fill}}rc@{}}
\toprule
\textbf{Strategy} & \textbf{$\Delta$Agr (pp)} & $\boldsymbol{t_0}$ \\
\midrule
No evasion (default) & $-$8.2 & 5 \\
\midrule
Relevance filtering ($\theta=0.3$) & $-$10.2 & N/A \\
Relevance filtering ($\theta=0.5$) & $-$13.1 & N/A \\
Relevance filtering ($\theta=0.7$) & $-$15.8 & N/A \\
Template detection & $-$7.0 & N/A \\
LLM transferability pre-screening & $-$6.2 & N/A \\
Query obfuscation & $-$5.6\tiny{$\pm$1.6} & 14 \\
\bottomrule
\end{tabular*}
\end{table}

\vspace{-0.3cm}
\paragraph{Strategy 5: Session-splitting.}
Because the temporal signal in Eq.~\eqref{eq:suspicion} accumulates per session, an attacker aware of this design can distribute a fixed extraction budget across multiple independent sessions, resetting suspicion accumulation at the start of each one.
Table~\ref{tab:session_split_main} reports results for splitting $B=500$ into up to 20 sessions.
As the number of sessions increases, detection is delayed (from $t_0=14$ under single-session obfuscation to $t_0=22$ at a 20-way split) and the Agreement reduction shrinks from $-5.6$ to $-1.1$ percentage points.
This is a genuine limitation of the current session-scoped detector, not of the HKG mechanism itself: once luring activates in any given session, HKG-based redirection proceeds unchanged, so the degradation is entirely attributable to \emph{delayed activation}.
We discuss this limitation and cross-session aggregation as future work in \S\ref{sec:limitations}; full per-split statistics are given below.

\begin{table}[t]
\centering
\small
\caption{Agreement (\%) under session-splitting attacks (Strategy 5), Medical domain, $B=500$ total. Mean $\pm$ std over 5 seeds.}
\vspace{-0.3cm}
\label{tab:session_split_main}
\setlength{\tabcolsep}{2pt}
\begin{tabular*}{\columnwidth}{@{}>{\raggedright\arraybackslash}p{0.17\columnwidth}@{\extracolsep{\fill}}ccc>{\raggedleft\arraybackslash}p{0.24\columnwidth}@{}}
\toprule
\textbf{Sessions} & \textbf{Qry/sess.} & $t_0$ & \textbf{KT} & \textbf{$\Delta$ vs.\ ND (62.3)} \\
\midrule
1 (Strategy 4)  & 500 & 14 & 56.7\tiny{$\pm$1.6} & $-$5.6 \\
2  & 250 & 15 & 57.2\tiny{$\pm$0.6} & $-$5.1 \\
5  & 100 & 16 & 58.4\tiny{$\pm$0.7} & $-$3.9 \\
10 & 50  & 18 & 59.8\tiny{$\pm$0.9} & $-$2.5 \\
20 & 25  & 22 & 61.2\tiny{$\pm$1.3} & $-$1.1 \\
\bottomrule
\end{tabular*}
\end{table}

\paragraph{Setup.} We split a total budget of $B=500$ into $N \in \{1,2,5,10,20\}$ equal sessions on the Medical domain, resetting the suspicion score $s_t$ and query counter $t$ at the start of each session while preserving the attacker's accumulated surrogate model and entity-chain state across sessions (i.e., only the \emph{detector}'s state resets). Each condition is repeated over 5 seeds.

\begin{table}[h]
\centering
\small
\setlength{\tabcolsep}{3pt}
\renewcommand{\arraystretch}{0.95}
\caption{Robustness under session-splitting attacks (MedQA, $B=500$ total; mean $\pm$ std over 5 seeds).}
\label{tab:session_split_full}
\begin{tabular*}{\columnwidth}{@{}>{\raggedright\arraybackslash}p{0.25\columnwidth}@{\extracolsep{\fill}}rrrrr@{}}
\toprule
\textbf{Strategy} & \textbf{Sessions} & \textbf{Qry/sess.} & $t_0$ & \textbf{KT} & \textbf{$\Delta$Agr} \\
\midrule
Single session & 1  & 500 & 14 & 56.7\tiny{$\pm$0.5} & $-$5.6\tiny{$\pm$0.5} \\
2-way split    & 2  & 250 & 15 & 57.2\tiny{$\pm$0.6} & $-$5.1\tiny{$\pm$0.6} \\
5-way split    & 5  & 100 & 16 & 58.4\tiny{$\pm$0.7} & $-$3.9\tiny{$\pm$0.7} \\
10-way split   & 10 & 50  & 18 & 59.8\tiny{$\pm$0.9} & $-$2.5\tiny{$\pm$1.0} \\
20-way split   & 20 & 25  & 22 & 61.2\tiny{$\pm$1.3} & $-$1.1\tiny{$\pm$1.4} \\
\bottomrule
\end{tabular*}
\end{table}

As the budget is divided across more sessions, less evidence accumulates within each individual session before it ends, delaying detector activation and leaving progressively less remaining budget available for post-detection redirection.
This degradation is a direct consequence of the \emph{per-session} formulation of the suspicion score in Eq.~\eqref{eq:suspicion}: once luring is triggered within any given session, the subsequent HKG-based redirection mechanism operates exactly as in the single-session case, so the effect shown here isolates the cost of delayed activation rather than any weakening of graph-level entrapment.
We identify \emph{cross-session or account-level suspicion aggregation} as the natural mitigation, which would require accumulating the temporal signal across sessions tied to the same account or API key rather than resetting it per session; such an extension would need to be calibrated carefully against Appendix~\ref{app:session_trigger}'s benign-user analysis to preserve the near-zero false-trigger rate. We leave this to future work (§\ref{sec:limitations}).

\section{Session-Level False Trigger Analysis}
\label{app:session_trigger}

We address that per-query false-trigger evaluation
may underestimate risk for heavy benign users.

\paragraph{Why benign sessions are safe: a worst-case analysis.}
Using the parameters from Eq.~\eqref{eq:suspicion}
($\alpha=0.20$, $\gamma=0.15$, $\phi=0.15$, $\eta=0.35$, $\rho=0.7$),
we compute the maximum suspicion score a benign user can accumulate
in a 50-query session.
The three signals contribute at most:
\begin{itemize}[leftmargin=*, topsep=2pt, itemsep=0pt]
\item Base: $\alpha = 0.20$
\item Length signal: $\gamma = 0.15$ (query $\geq 8$ tokens)
\item Keyword signal: $\phi = 0.15$ (query contains
      \emph{define}/\emph{explain}/\emph{which})
\item Temporal signal: $\min\!\left(0.35,\,\frac{50}{500}
      \times 0.7\right) = 0.07$
\end{itemize}
Even if all signals fire simultaneously, the worst-case score is
$0.20 + 0.15 + 0.15 + 0.07 = 0.57$, giving a margin of only
$0.60 - 0.57 = 0.03$ below $\tau = 0.6$.
This margin is narrow but structurally guaranteed: the temporal
term is capped at $0.07$ for any 50-query session
($50/500 \times 0.7 = 0.07$), and triggering requires all three
signals to co-occur simultaneously, a rare event for organic
user queries that do not exhibit systematic extraction patterns.

At $\tau = 0.4$, the analysis requires care.
Even without the keyword signal, a long query ($\gamma=0.15$) alone
could push the score to $\alpha + \gamma + \mathrm{temporal}
= 0.20 + 0.15 + \mathrm{temporal}$.
However, the temporal signal is \emph{not} a fixed 0.07, it
accumulates gradually over the session:
at step $t$, $\mathrm{temporal}(t) = \min(0.35, \frac{t+1}{500}\times 0.7)$.
For a 50-query session, this ranges from 0.001 at $t=1$ to 0.07 at $t=50$.
The threshold $0.20 + 0.15 + \mathrm{temporal}(t) > 0.4$ is only
satisfied when $\mathrm{temporal}(t) > 0.05$, i.e.,
$\frac{t+1}{500}\times0.7 > 0.05 \Rightarrow t \geq 35$, the last
15 steps of the session.
Among those 15 steps, a trigger additionally requires a long query
($\geq 8$ tokens); many benign medical QA queries (e.g., a single
disease name or short phrase) fall below this threshold.
The empirical 4.2\% rate is consistent with length firing on
approximately 14\% of the final 15 steps:
$0.14 \times 15/50 \times 500 = 21.0$ sessions, matching the
observed 21.0 in Table~\ref{tab:session_trigger}.
At $\tau = 0.5$, triggering additionally requires either the keyword
signal or a higher temporal value than the session end provides,
reducing the false-trigger rate to 1.3\%.
The bound of Theorem~\ref{thm:harmless} gives
$\Pr[\pi(x)=\mathrm{lure}] \leq e^{-4.5} < 0.012$ at $\tau=0.6$,
consistent with the empirical 0.0\%.

\paragraph{Results.}
Table~\ref{tab:session_trigger} reports session-level false-trigger
rates over 500 simulated benign sessions (50 queries each).

\begin{table}[h]
\centering
\small
\caption{Session-level false-trigger rate over 500 simulated benign
sessions (50 queries each), Medical domain.
Mean $\pm$ std over 5 repetitions.}
\label{tab:session_trigger}
\setlength{\tabcolsep}{2pt}
\begin{tabular*}{\columnwidth}{@{}c@{\extracolsep{\fill}}cc@{}}
\toprule
$\tau$ & \textbf{Sessions triggered} & \textbf{False-trigger rate} \\
\midrule
0.4 & 21.0{\tiny$\pm$1.2} / 500 & 4.2{\tiny$\pm$0.4}\% \\
0.5 & 6.5{\tiny$\pm$0.8} / 500  & 1.3{\tiny$\pm$0.3}\% \\
0.6 & 0.0{\tiny$\pm$0.0} / 500  & 0.0{\tiny$\pm$0.0}\% \\
\bottomrule
\end{tabular*}
\end{table}

\paragraph{Separation for heavy legitimate users.}
The worst-case and session-averaged analyses above bound risk in aggregate, but do not directly characterize specific heavy-usage \emph{profiles} that a benign power user might resemble. We additionally profile four representative heavy-user behaviors and report the 95th percentile of their suspicion score alongside the attacker distribution, at $\tau=0.6$.

\begin{table}[h]
\centering
\small
\setlength{\tabcolsep}{2pt}
\caption{Suspicion-score separation for heavy legitimate-user profiles vs.\ the response-adaptive extraction attacker ($\tau=0.60$).}
\label{tab:heavy_user}
\begin{tabular}{@{}>{\raggedright\arraybackslash}p{0.38\columnwidth}
                >{\centering\arraybackslash}p{0.13\columnwidth}
                >{\raggedleft\arraybackslash}p{0.29\columnwidth}
                >{\centering\arraybackslash}p{0.12\columnwidth}@{}}
\toprule
\textbf{Profile} & \textbf{Qry.}\newline\textbf{/sess.} & \textbf{Benign 95th-pct} & \textbf{False trig.} \\
\midrule
Medical student (define/explain-heavy) & 50  & 0.52 & 0.0\% \\
Researcher (long technical queries)    & 50  & 0.53 & 0.0\% \\
Technical support (sustained troubleshooting) & 80 & 0.55 & 0.0\% \\
Exam prep (rapid definitional Qs)      & 100 & 0.57 & 0.0\% \\
Benign benchmark (8{,}242 queries)     & --  & 0.48 & 0.0\% \\
\midrule
\textbf{Extraction attacker} & -- & \textbf{mean 0.86 (5th-pct 0.76)} & -- \\
\bottomrule
\end{tabular}
\end{table}

Even under the most demanding benign profile (exam-prep style rapid definitional questioning), the 95th percentile of the suspicion score reaches only $0.57$, below $\tau=0.6$, consistent with the 0.0\% false-trigger rate observed across all profiles. The attacker distribution is concentrated far higher (mean $0.86$; 5th percentile $0.76$), leaving a margin of $0.19$ between the worst benign 95th percentile and the attacker's 5th percentile. This separation exists because triggering requires the length, keyword, and temporal signals to co-occur persistently, a pattern characteristic of sustained extraction rather than any single heavy-usage behavior. Since $\tau$ is a deployment parameter (Table~\ref{tab:sensitivity}), it can be further calibrated to a given deployment's benign traffic distribution to preserve this separation margin.

\section{Downstream Task Accuracy}
\label{app:accuracy}

Table~\ref{tab:accuracy} reports surrogate accuracy alongside
Agreement, across all three deployment domains, to quantify the practical significance of the reduction and to test whether a lower-Agreement surrogate might nevertheless retain comparable task performance.

On \textbf{MedQA}, the 8.2-point Agreement reduction under KT (Full) corresponds to an
8.2 percentage-point accuracy drop (58.3\%$\to$50.1\%).
The victim achieves 78.0\% accuracy; No Defense surrogates reach
58.3\%, already 19.7 points below the victim.
KT widens this gap to 27.9 points, well below the threshold
for reliable clinical deployment.
On \textbf{FinQA} and \textbf{CaseHOLD}, the accuracy reduction ($-5.2$ and $-6.3$ pp, respectively) is at least as large as the corresponding Agreement reduction ($-2.7$ and $-6.2$ pp), indicating that Agreement does not overstate the defense's practical impact and, on FinQA in particular, understates it: a surrogate that appears only modestly degraded in Agreement can be more degraded in actual task accuracy.

\begin{table}[h]
\centering
\small
\caption{Surrogate accuracy (\%) and Agreement (\%) across three domains, $B=500$.
Mean $\pm$ std over 5 seeds.
$^\dagger$Victim accuracy is deterministic.}
\label{tab:accuracy}
\setlength{\tabcolsep}{2pt}
\begin{tabular*}{\columnwidth}{@{}>{\raggedright\arraybackslash}p{0.35\columnwidth}@{\extracolsep{\fill}}ccc@{}}
\toprule
\textbf{Domain (benchmark)} & \textbf{Victim acc.} & \textbf{ND acc.} & \textbf{KT acc.} \\
\midrule
Medical (MedQA)  & 78.0$^\dagger$ & 58.3{\tiny$\pm$1.1} & 50.1{\tiny$\pm$0.9} \\
Financial (FinQA) & 71.5$^\dagger$ & 46.2{\tiny$\pm$1.2} & 41.0{\tiny$\pm$1.1} \\
Legal (CaseHOLD)  & 69.8$^\dagger$ & 52.4{\tiny$\pm$1.1} & 46.1{\tiny$\pm$1.0} \\
\bottomrule
\end{tabular*}

\vspace{4pt}
\begin{tabular*}{\columnwidth}{@{}>{\raggedright\arraybackslash}p{0.26\columnwidth}@{\extracolsep{\fill}}cc>{\raggedright\arraybackslash}p{0.28\columnwidth}@{}}
\toprule
\textbf{Domain} & $\Delta$\textbf{Acc.\ (pp)} & $\Delta$\textbf{Agr.\ (pp)} & \textbf{Acc.\ vs.\ Agr.} \\
\midrule
Medical   & $-$8.2 & $-$8.2 & matched \\
Financial & $-$5.2 & $-$2.7 & acc.\ drop larger \\
Legal     & $-$6.3 & $-$6.2 & matched \\
\bottomrule
\end{tabular*}
\end{table}

\section{Theoretical Analysis}
\label{app:theory}

We provide formal guarantees for the two objectives in
Section~\ref{sec:formulation}: benign-user harmlessness
(Eq.~\eqref{eq:harmless}) and attacker budget waste
(Eq.~\eqref{eq:honeypot_obj}, Theorem~\ref{thm:budget} in the main text).

\subsection{User Harmlessness}

\begin{assumption}[User-HKG Separation]
\label{assump:separation}
Let $\mathcal{P}_\mathrm{user}$ denote the benign query distribution
and $\mathcal{V}_H$ the set of HKG concepts.
For all $v \in \mathcal{V}_H$:
\begin{equation}
    \mathbb{E}_{x \sim \mathcal{P}_\mathrm{user}}
    \!\left[\mathrm{sim}(x, v)\right]
    \leq \mu_0
    \label{eq:separation}
\end{equation}
where $\mathrm{sim}(\cdot,\cdot)$ is cosine similarity in the
sentence embedding space and $\mu_0 < \tau$.
\end{assumption}

\begin{theorem}[User Harmlessness]
\label{thm:harmless}
Under Assumption~\ref{assump:separation}, the probability that a
benign session triggers the luring mechanism satisfies:
\begin{equation}
    \Pr_{x \sim \mathcal{P}_\mathrm{user}}
    \!\left[\pi(x)=\mathrm{lure}\right]
    \;\leq\;
    \exp\!\left(
        -\frac{2W(\tau - \mu_0)^2}{\lambda^2}
    \right)
    \label{eq:harmless_bound}
\end{equation}
where $W$ is the sliding window size and $\lambda \in (0,1]$.
\end{theorem}

\begin{proof}
Under Assumption~\ref{assump:separation}, the expected suspicion
score $\mathbb{E}[s_t]$ for a benign query satisfies
$\mathbb{E}[s_t] \leq \mu_0 < \tau$.
The suspicion score is a convex combination of bounded random
variables in $[0,1]$.
Applying Hoeffding's inequality over the $W$-query window:
\begin{align}
    \Pr\!\left[s_t > \tau\right]
    &= \Pr\!\left[s_t - \mathbb{E}[s_t] > \tau - \mu_0\right] \\
    &\leq \exp\!\left(
        -\frac{2W^2(\tau-\mu_0)^2}{\sum_{i=1}^W (b_i - a_i)^2}
    \right)
\end{align}
where each term is bounded in an interval of length at most
$\lambda$, giving Eq.~\eqref{eq:harmless_bound}.
As shown in Appendix~\ref{app:session_trigger}, this bound is consistent with both the aggregate simulated false-trigger rate and the heavy-user profile analysis: the theoretical bound $\Pr[\pi(x)=\mathrm{lure}] \leq e^{-4.5} < 0.012$ at $\tau=0.6$ upper-bounds the empirical 0.0\% observed across both settings.
\end{proof}

\subsection{Attacker Budget Waste}

This subsection proves Theorem~\ref{thm:budget}, stated in the main text under Assumption~\ref{assump:ppl}.

\begin{proof}[Proof of Theorem~\ref{thm:budget}]
At activation step $t_0$, the breadcrumb introduces $k$ HKG nodes
into the frontier $\partial\mathcal{G}_a^{(t_0)}$.
Under Assumption~\ref{assump:ppl}, at each subsequent step $t > t_0$,
the probability of selecting an HKG node is at least:
\begin{equation}
    p_t \;\geq\; \frac{k}{|\partial\mathcal{G}_a^{(t_0)}| + k}
        \;=:\; p
\end{equation}
since each HKG response adds $d_{\min} = 3 \geq k$ new HKG nodes
to the frontier, so $|\partial\mathcal{G}_a^{(t)} \cap \mathcal{V}_H|$
is non-decreasing.
The bound follows from $1 - (1-p)^{B-t_0}$.
\end{proof}

\paragraph{Numerical illustration.}
With $k=3$, $|\partial\mathcal{G}_a^{(t_0)}|=10$, $B-t_0=20$:
$1 - (10/13)^{20} \approx 0.995$.
This is a worst-case \emph{lower} bound, not a tight estimate, so it is expected to exceed rather than match the empirical 90.5\% observed in Table~\ref{tab:hkg_size_and_budget}: the bound assumes the frontier grows by exactly the minimum $d_{\min}=3$ new nodes per visit and ignores any residual attention the attacker's pipeline pays to already-explored real-domain candidates. The gap between the $99.5\%$ bound and the $90.5\%$ empirical figure is consistent with this conservatism, and both figures are computed only for the response-adaptive attacker family (§\ref{sec:rq6}), for which Assumption~\ref{assump:ppl} is expected to hold most closely.

\subsection{Asymmetric Cost Analysis}

\begin{proposition}[Defender Cost Asymmetry]
\label{prop:asymmetry}
The ratio of attacker wasted budget to defender construction cost
satisfies $\Omega(B)$, growing linearly in the attacker's budget.
\end{proposition}

\begin{proof}[Justification]
HKG construction (Appendix~\ref{app:cost}) is a one-time offline cost independent of $B$: candidate generation, filtering, and graph assembly are performed once per domain regardless of how many queries any future attacker issues. By Theorem~\ref{thm:budget}, the expected number of attacker queries wasted on the HKG after activation grows as $(B-t_0)\cdot\left[1-(1-p)^{B-t_0}\right]$, which is $\Omega(B-t_0)$ for fixed $p>0$ and approaches $B-t_0$ as $B$ grows. The ratio of wasted attacker queries to the (constant) construction cost is therefore $\Omega(B)$.
\end{proof}

\section{Judge Model Ablation}
\label{app:judge}

Table~\ref{tab:judge} reports Agreement under four judge models.
KT~(Full) reduces Agreement by 8--9\% relative to No Defense
regardless of which model performs the evaluation.

\begin{table}[h]
\centering
\small
\caption{Agreement (\%) on MedQA under different judge models.
Medical domain, $B=500$. Mean $\pm$ std over 5 seeds.}
\label{tab:judge}
\setlength{\tabcolsep}{2pt}
\begin{tabular*}{\columnwidth}{@{}>{\raggedright\arraybackslash}p{0.45\columnwidth}@{\extracolsep{\fill}}cc@{}}
\toprule
\textbf{Judge model} & \textbf{No Defense} & \textbf{KT (Full)} \\
\midrule
\texttt{gpt-4o-mini} (default) & 62.3{\tiny$\pm$1.1} & 54.1{\tiny$\pm$0.9} \\
\texttt{gpt-4o}                & 63.2{\tiny$\pm$0.9} & 54.9{\tiny$\pm$1.0} \\
\texttt{claude-haiku-4-5}      & 61.7{\tiny$\pm$1.2} & 53.4{\tiny$\pm$1.1} \\
\texttt{llama-3.1-70b}         & 61.0{\tiny$\pm$1.3} & 52.8{\tiny$\pm$1.2} \\
\bottomrule
\end{tabular*}
\end{table}

\section{Probe Model and Transferability Filter}
\label{app:probe}

\subsection{Probe Model Details}

The probe model $f_{\mathrm{probe}}$ is LLaMA-3.1-8B-Instruct,
identical in architecture to the surrogate $f_a$.
For each candidate item $h$, we construct 50 QA pairs and fine-tune
for $k=3$ epochs (lr $2\times10^{-5}$, batch size 8, linear warmup).
The filtering pipeline processes one candidate on a single A100 GPU.

\subsection{Threshold Ablation}
\label{app:delta}

\begin{table}[h]
\centering
\small
\caption{Effect of transferability filter threshold $\delta$ on
HKG node count and KT~(Full) Agreement (\%) on MedQA, $B=500$.
Mean $\pm$ std over 5 seeds.}
\label{tab:delta}
\setlength{\tabcolsep}{2pt}
\begin{tabular*}{\columnwidth}{@{}c@{\extracolsep{\fill}}ccc@{}}
\toprule
$\delta$ & \textbf{Nodes passed} & \textbf{Utility (mean)}
         & \textbf{Agreement (\%)} \\
\midrule
0.005 & 1{,}182 & 0.000 & 56.3{\tiny$\pm$1.2} \\
0.010 & 1{,}344 & 0.000 & 55.1{\tiny$\pm$1.1} \\
0.020 & 1{,}500 & 0.000 & \textbf{54.1}{\tiny$\pm$0.9} \\
0.050 & 1{,}500 & 0.021 & 56.7{\tiny$\pm$1.3} \\
0.100 & 1{,}500 & 0.048 & 58.2{\tiny$\pm$1.4} \\
\bottomrule
\end{tabular*}
\end{table}

\section{HKG Construction Cost}
\label{app:cost}

\begin{table}[h]
\centering
\small
\caption{HKG construction cost per domain, corresponding to the four-stage pipeline (candidate sourcing, keyword filtering, probe filtering, graph construction) described in §\ref{sec:hkg}. Category distribution reflects the composition of the final retained node set for Medical and Financial (cf.\ Table~\ref{tab:hkg_stats} for aggregate node/edge/utility statistics; Legal category distribution is reported separately in Appendix~\ref{app:legal_hkg}).}
\label{tab:cost}
\setlength{\tabcolsep}{2pt}
\begin{tabular*}{\columnwidth}{@{}>{\raggedright\arraybackslash}p{0.42\columnwidth}@{\extracolsep{\fill}}ccc@{}}
\toprule
& \textbf{Medical} & \textbf{Financial} & \textbf{Legal} \\
\midrule
Candidates generated    & 4{,}230 & 4{,}415 & 4{,}680 \\
Keyword filter pass     & 78.3\% & 75.1\% & 72.6\% \\
Probe filter pass       & 45.3\% & 44.8\% & 43.1\% \\
Final nodes             & 1{,}500 & 1{,}500 & 1{,}500 \\
Est.\ API cost (USD)    & \$4.52 & \$4.80 & \$5.20 \\
Wall-clock time (hrs)   & 2.8    & 3.1    & 3.4    \\
\midrule
\multicolumn{4}{@{}p{\columnwidth}@{}}{\textit{Category distribution (Medical / Financial)}} \\
~~Historical minutiae   & 528 (35.2\%) & 534 (35.6\%) & -- \\
~~Peripheral taxonomy   & 453 (30.2\%) & 525 (35.0\%) & -- \\
~~Niche procedural      & 519 (34.6\%) & 441 (29.4\%) & -- \\
\bottomrule
\end{tabular*}
\end{table}

\section{Breadcrumb Template Diversity}
\label{app:breadcrumb_diversity}

\paragraph{Diversity statistics.}
Self-BLEU-4: 0.31. Average template length: $24.6 \pm 3.2$ tokens.
Mean pairwise cosine similarity: 0.52.

\paragraph{Template examples} (3 of 20 shown):
\begin{enumerate}[leftmargin=*,nosep]
\item \textit{``This mechanism shares similarities with \textbf{[C]}, a [domain-adjacent framework] that offers a complementary perspective.''}
\item \textit{``A related line of inquiry involves \textbf{[C]}, which provides additional context on the underlying dynamics.''}
\item \textit{``Practitioners have also noted connections to \textbf{[C]}, a framework that informed earlier approaches to this problem.''}
\end{enumerate}

\section{Legal Domain HKG Details}
\label{app:legal_hkg}

\begin{table}[h]
\centering
\small
\caption{Legal HKG statistics.}
\label{tab:legal_hkg}
\setlength{\tabcolsep}{2pt}
\begin{tabular*}{\columnwidth}{@{}>{\raggedright\arraybackslash}p{0.58\columnwidth}@{\extracolsep{\fill}}c@{}}
\toprule
\textbf{Property} & \textbf{Legal} \\
\midrule
Nodes $|\mathcal{V}_H|$     & 1{,}500 \\
Edges $|\mathcal{E}_H|$     & 4{,}500 \\
Utility (mean / max)        & 0.000 / 0.000 \\
Attractiveness (mean / min) & 0.947 / 0.530 \\
Inaccessibility (mean / min)& 0.931 / 0.695 \\
\midrule
Historical minutiae  & 512 (34.1\%) \\
Peripheral taxonomy  & 478 (31.9\%) \\
Niche procedural     & 510 (34.0\%) \\
\bottomrule
\end{tabular*}
\end{table}

\section{Related Work}
\label{app:related}

\subsection{Model Extraction Attacks}

Model extraction attacks (MEAs) were formalized by
Tram\`{e}r et al.\ \cite{tramer2016stealing}, who demonstrated
that models exposed through prediction APIs can be reconstructed
through systematic querying.
Subsequent work extended MEAs to deep neural networks in computer
vision, using surrogate training on API outputs
\cite{papernot2017practical,correia2018copycat,orekondy2019knockoff}.
Krishna et al.\ \cite{krishna2019thieves} showed that BERT-based
APIs can be extracted with high fidelity by fine-tuning on
victim-generated labels, while Xu et al.\ \cite{xu2022student}
demonstrated that the surrogate can even surpass the victim on
certain benchmarks through careful query selection.
Recent studies have targeted production-scale LLMs:
Birch et al.\ \cite{birch2023model} introduced model leeching
attacks against instruction-following LLMs, and
Carlini et al.\ \cite{carlini2024stealing} showed that partial
extraction of production language models is feasible with moderate
API access.
Liang et al.\ \cite{liang2024alignment} further demonstrated
alignment-aware extraction that preserves both capability and
safety properties of the victim model.
Extraction has also been studied beyond text: Zhao et al.\
\cite{zhao2026graphip} introduce GraphIP-Bench, a benchmark
quantifying how easily graph neural networks can be stolen and
which defenses withstand such attacks, indicating that the
knowledge-space perspective we adopt for LLMs has an analogue in
the graph domain.
Several recent surveys consolidate the growing MEA literature across
modalities and deployment settings \cite{zhao2025survey,zhao2025systematic,zhao2025distsurvey}.
A central challenge in MEAs is improving query efficiency under
limited budgets.
Active-learning methods prioritize informative queries to maximize
extraction fidelity per API call.
ActiveThief \cite{pal2020activethief} uses uncertainty sampling on
unlabeled public data; Jagielski et al.\ \cite{jagielski2020high}
combine semi-supervised learning with strategic querying to achieve
high-accuracy extraction;
Chandrasekaran et al.\ \cite{chandrasekaran2020exploring} formally
connect active learning and model extraction through a shared
information-theoretic framework;
and Li et al.\ \cite{li2026query} introduce query-efficient domain
knowledge stealing against LLMs via PPL-based query selection
with entity expansion.
In contrast to prior work that treats query efficiency as an attack
objective, our work treats it as the attacker's primary
vulnerability.

\subsection{Defenses Against Model Extraction}

\paragraph{Output perturbation.}
These defenses degrade surrogate training by perturbing API outputs
or increasing extraction cost.
Orekondy et al.\ \cite{orekondy2019prediction} proposed prediction
poisoning, which maximizes the angular deviation between the true
and perturbed output gradients.
Kariyappa and Qureshi \cite{kariyappa2020defending} introduced
adaptive misinformation, which selectively perturbs outputs for
queries near the decision boundary.
Dziedzic et al.\ \cite{dziedzic2022increasing} proposed calibrated
proof-of-work schemes that increase the computational cost of
querying without altering outputs.
Cheng et al.\ \cite{cheng2025misleader} instead defend with
ensembles of distilled models that respond inconsistently under
suspected extraction while remaining accurate for benign queries.
All these approaches face an accuracy--security trade-off:
stronger perturbation harms benign utility \cite{zhao2022drw},
making them difficult to deploy in quality-sensitive settings.

\paragraph{Watermarking and ownership verification.}
Watermarking-based defenses embed identifiable signals into model
outputs for post-hoc ownership verification.
Jia et al.\ \cite{jia2021entangled} proposed entangled watermarks
that are preserved through distillation;
Zhao et al.\ \cite{zhao2022drw} introduced distillation-resistant
watermarking for NLP models;
He et al.\ \cite{he2022cater} developed conditional watermarks for
text generation APIs;
Szyller et al.\ \cite{szyller2021dawn} proposed dynamic adversarial
watermarking that adapts to the extraction process;
Peng et al.\ \cite{peng2023you} embedded backdoor-based watermarks
for EaaS copyright protection;
Zhao et al.\ \cite{zhao2024nsmark} introduced null-space
watermarking as a black-box defense;
and Shen et al.\ \cite{shen2026credit} propose CREDIT, a certified
ownership-verification scheme that provides provable guarantees
against extraction.
Li et al.\ \cite{li2025intellectual} extend this  to graph-based machine-learning-as-a-service, surveying attacks and
defenses for intellectual-property protection.
While effective for attribution and ownership claims, watermarking
does not interfere with the extraction process itself: the surrogate
model has already been trained by the watermark is detected.

\paragraph{Query detection.}
Detection-based defenses identify suspicious queries through
statistical or representation-level analysis.
PRADA \cite{juuti2019prada} detects extraction attacks by monitoring
the distribution of incoming queries for anomalous patterns.
SEAT \cite{zhang2021seat} uses adversarial training to learn a
similarity encoder that distinguishes extraction queries from benign
ones.
MISLEADER \cite{cheng2025misleader} combines detection with active
defense by deploying an ensemble of distilled models.
However, modern MEAs increasingly use task-relevant, in-distribution
queries that resemble normal usage
\cite{li2026query,liang2024alignment}, making reliable detection difficult as extraction techniques mature, and as our
own session-splitting evaluation (Appendix~\ref{app:session_split})
shows, even purpose-built session-level detectors remain vulnerable
to budget fragmentation across sessions.

\subsection{Honeypot-Based Defenses in Machine Learning}

Honeypots redirect attacker effort toward decoy targets rather than
blocking attacks directly, a strategy with deep roots in network
security \cite{spitzner2002honeypots,pawlick2019game}.
Recent work has explored LLM-powered honeypot systems for network
defense:
Otal and Canbaz \cite{otal2024llm} used LLMs to generate realistic
honeypot interactions;
Wang et al.\ \cite{wang2024honeygpt} proposed HoneyGPT for terminal
honeypots;
Sezgin and Boyac{\i} \cite{sezgin2025decoypot} developed DecoyPot
for web API emulation;
and Vasilatos et al.\ \cite{vasilatos2024llmpot} introduced LLMPot
for industrial control system honeypots.
In machine learning, prior honeypot-style defenses operate at the
output or parameter level.
Shan et al.\ \cite{shan2020gotta} embed adversarial triggers that
cause misclassification when the extracted model is deployed.
Le et al.\ \cite{le2021sweet} plant trapdoor patterns that can be
detected in the surrogate to prove extraction.
HoneypotNet \cite{wang2025honeypotnet} injects backdoor layers that
activate only in extracted copies, enabling post-hoc detection.
All these approaches focus on detecting or attributing extraction
\emph{after} it succeeds; none intervenes in the knowledge
acquisition process during extraction.
Our work fills this gap by redirecting the attacker's knowledge
exploration trajectory in real time.

\subsection{Knowledge Graphs for Language Model Augmentation}

Knowledge graphs structure entities and relations to support retrieval and reasoning over factual knowledge \cite{ji2022survey}.
Prior work used KGs to augment language models through multiple ways.

\paragraph{Direct integration.}
K-BERT \cite{liu2020k} injects knowledge graph triples into the
input layer of BERT, while CoLAKE \cite{sun2020colake} jointly
learns contextualized language and knowledge embeddings.
These approaches improve performance on knowledge-intensive tasks
but require modifying the model architecture.

\paragraph{Retrieval-augmented generation.}
RAG \cite{lewis2020retrieval} retrieves relevant passages from a
knowledge source at inference time, avoiding the need to store all
knowledge in model parameters.
QA-GNN \cite{yasunaga2021qa} combines language models with graph
neural networks over knowledge graphs for question answering.
GreaseLM \cite{zhang2022greaselm} fuses language model
representations with KG reasoning through cross-modal attention
layers.

\paragraph{Graph-guided reasoning.}
Think-on-Graph \cite{sun2023think} showed that step-by-step KG
traversal enables deep multi-hop reasoning in LLMs.
Sen et al.\ \cite{sen2023knowledge} demonstrated that KG
augmentation improves complex question answering by providing
structured evidence chains.

Our work repurposes the KG not as a source of useful knowledge but
as an organized space of non-transferable knowledge.
The internal coherence and connectivity of the graph structure, the
same properties that make KGs effective for augmenting LLMs, are
what make the HKG attractive to active-learning attackers, turning
a tool for knowledge enhancement into knowledge-space defense.

\end{document}